\documentclass{llncs}

\usepackage{geometry}
\usepackage{algorithm}
\usepackage{algpseudocode}
\usepackage{array, makecell}
\usepackage{subfig}
\usepackage{graphicx}   

\usepackage{glossaries}
\newacronym{cps}{CPS}{Cyber-Physical System}
\newacronym{ids}{IDS}{Intrusion Detection System}
\newacronym{idss}{IDSs}{Intrusion Detection Systems}
\newacronym{iot}{IoT}{Internet of Things}
\newacronym{wsn}{WSN}{Wireless Sensor Network}
\newacronym{ecc}{ECC}{Elliptic Curve Cryptography}
\newacronym{sgw}{SGW}{SinGle chaining Watermark}
\newacronym{lsb}{LSB}{Least Significant Bit}
\newacronym{lwc}{LWC}{Lightweight Chained Watermarking}
\newacronym{fwcd}{FWC-D}{lightweight fragile watermarking scheme}
\newacronym{rwfs}{RWFS}{Randomized Watermarking Filtering Scheme}
\newacronym{mac}{MAC}{Medium Access Control}
\newacronym{prw}{PRW}{position random watermark}
\newacronym{pufs}{PUFs}{physical unclonable functions}
\newacronym{rssi}{RSSI}{Received Signal Strength Indicator}
\newacronym{act}{ACT}{Asymmetric Cryptography Technique}
\newacronym{ddos}{DDoS}{Distributed Denial of Service}
\newacronym{zwt}{ZWT}{Zero-Watermarking Scheme}
\newacronym{aes}{AES}{Advanced Encryption Standard}
\newacronym{des}{DES}{Data Encryption Standard}

\newcommand{\ie}{i.\,e.,~}
\newcommand{\eg}{e.\,g.,~}

\begin{document}

\title{ZIRCON: Zero-watermarking-based approach for data integrity and secure provenance in IoT networks}

%% \subtitle{~~\\Analysis of a Practical Case for Postal Applications}

\author{Omair Faraj\inst{1,2} \and David Meg\'ias  \inst{1}  \and Joaquin Garcia-Alfaro \inst{2}}

\institute{
Internet Interdisciplinary Institute (IN3), Universitat Oberta de Catalunya (UOC), CYBERCAT-Center for Cybersecurity Research of Catalonia, Barcelona, Spain
\and SAMOVAR, T\'el\'ecom SudParis, Institut Polytechnique de Paris, Palaiseau, France.
}

\maketitle

\begin{abstract}
The \gls*{iot} is integrating the Internet and smart devices in almost every domain such as home automation, e-healthcare systems, vehicular networks, industrial control and military applications. In these sectors, sensory data, which is collected from multiple sources and managed through intermediate processing by multiple nodes, is used for decision-making processes. Ensuring data integrity and keeping track of data provenance is a core requirement in such a highly dynamic context, since data provenance is an important tool for the assurance of data trustworthiness. Dealing with such requirements is challenging due to the limited computational and energy resources in \gls*{iot} networks. This requires addressing several challenges such as processing overhead, secure provenance, bandwidth consumption and storage efficiency. In this paper, we propose ZIRCON, a novel zero-watermarking approach to establish end-to-end data trustworthiness in an \gls*{iot} network. 
In ZIRCON, provenance information is stored in a tamper-proof centralized network database through watermarks, generated at source node before transmission. We provide an extensive security analysis showing the resilience of our scheme against passive and active attacks. We also compare our scheme with existing works based on performance metrics such as computational time, energy utilization and cost analysis. The results show that ZIRCON is robust against several attacks, lightweight, storage efficient, and better in energy utilization and bandwidth consumption, compared to prior art.\\

\end{abstract}

\section{Introduction}
\label{sec:intro}

The \gls*{iot} is an intelligent system composed of physical objects interconnected in a dynamic network infrastructure, which allows it to collect and exchange data between different sources and destinations~\cite{AMMAR2018,BOTTA2016}. These objects route data being captured from the environment to the unsafe Internet to be managed, processed and analyzed using different technologies. This makes it easier for an attacker to access the network and, thus, the system becomes vulnerable to intrusions~\cite{GUBBI2013}. Such intelligent systems are being used in various applications, such as healthcare systems, home appliances, car automation, industrial control, or environmental monitoring, among others~\cite{Lazarescu2017,Manikandan2018}. For these reasons, and for more than two decades, protecting and securing networks and information systems have been delivered through \gls*{idss}. It is difficult to apply traditional protection techniques (\eg cryptography, signature-based techniques, etc.) to \gls*{iot} networks due to some characteristics such as specific protocol stacks, constrained-resource devices, computational and power capabilities, storage limitations and network standards~\cite{SICARI2015}. At first, processing capabilities and storage limitations of network devices that host \gls*{ids} algorithms is a critical issue. In conventional networks, \gls*{ids} agents are deployed by network administrator in intermediate nodes that have high computing capacity ~\cite{zarpelao2017}. On the other hand, network nodes in \gls*{iot} environments are resource-constrained. Therefore, finding nodes with the ability to support \gls*{ids} agents is difficult in such systems. 

Another major issue is the network architecture. Specific nodes, such as routers and switches, are responsible for forwarding packets to final destinations that are directly connected to end systems in traditional networks. \gls*{iot} networks are generally multi-hop. In this case, network nodes forward packets simultaneously and act as end systems. Hence, for the sensed data packets to reach the final destination (\eg gateway, central processing unit, etc.) it will be forwarded through a path of sensor nodes placed on different light poles~\cite{zarpelao2017}. Sharing these data packets through the shared wireless medium expose the network to be vulnerable to several types of security attacks, such as data forgery, packet replay, data modification, data insertion, or packet drop attack~\cite{REN2008,KHAN2018}.

Consequently, there is a need for a lightweight \gls*{ids} scheme that maintains data integrity, trustworthiness and the ability to secure provenance to ensure that data is forwarded safely. Data provenance provides history of the data origin and how it is routed over time, which makes it an important tool for the assurance of data trustworthiness~\cite{Sultana2013,Dai2008}. Most of the previous research on provenance considered studying modeling, collecting and querying data provenance without focusing on its security. Moreover, very few approaches considered provenance in sensor networks. In such networks, there is a set of challenges to deploy provenance solutions. These challenges are (i) manage processing overhead of each individual node, (ii) efficiently transmit provenance while minimizing the additional bandwidth consumption, and (iii) transmit provenance securely from source to destination with the prompt react to any attack~\cite{Sultana2013}. 

In this paper, we provide a complete framework that deals with the above mentioned provenance challenges while ensuring data integrity in an \gls*{iot} network. These requirements can be achieved through watermarking techniques that are lightweight and require less computational and storage capabilities. Scalability is also an important requirement, since the size of provenance increases proportionally as the number of nodes engaged in the forwarding process increases. To address this issue, we introduce a centralized network database, connected to all nodes and gateway, that stores the provenance information at each hop. In this paper, we propose a zero-watermarking approach that securely communicates provenance information through a tamper-proof database and provides data integrity for real-time systems. A watermark is generated at the source \gls*{iot} device from the provenance information (IP address, data packet sensed time or received time) and a hash value of the data payload. The node stores this watermark in a tamper-proof centralized database and embeds it with the data packet to be sent through transmission channel. At intermediate nodes, the watermark is re-generated from original data for verification procedure. Finally, at the gateway, data integrity is verified through watermark re-generation process and the stored watermarks forming the complete provenance information are queried for validating provenance of data and constructing the data path. Hence, the proposed framework provides secure provenance transmission, scalability, secure transmission of watermarks and an in-route data integrity at each point in the network. 

The main contributions of this paper are summarized below:

\begin{itemize}
\item We propose a novel zero-watermarking scheme for \gls*{iot} networks to ensure data integrity in single- and multi-hop scenarios.
\item We provide a solution for the problem of secure provenance transmission in \gls*{iot} networks that is based on a zero-watermarking approach with a tamper-proof network database.
\item We analyze the security of our approach under two main adversary models: passive and external adversary.
\item We evaluate the performance of our approach with respect to some related techniques reported in the related literature. 
\end{itemize}

\noindent \textbf{Paper Organization ---} Section~\ref{sec:relatedWork} provides some preliminary background and surveys related work. Section~\ref{sec:sysmodel} presents the proposed zero-watermarking approach and solutions. Section~\ref{sec:security} analyzes the security of the proposed scheme. Section~\ref{sec:results} describes simulation results and analysis to validate our approach. Section~\ref{sec:discussion} provides discussion suggests some directions for future research. Finally, Section~\ref{sec:conclusion} concludes the paper.

\section{Background and Related Work}
\label{sec:relatedWork}

Digital watermarking is one of the well known advances in \gls*{wsn} security. It can detect if sensory data have been modified in a precise way, and prevents the interception of this data effectively. Additionally, it can be used for protecting content integrity of multimedia digital works as images, audio and video, and copyright information~\cite{van1994}. Digital watermarking has many advantages over other security techniques such as: 
\begin{enumerate}
    \item The three watermarking processes (generation, embedding and extraction) requires less energy than traditional encryption due to its light\-weight calculations.
    \item Carrier data directly holds watermark information without adding any network communication overhead~\cite{fei2006}.
    \item Digital-watermarking reduces end-to-end delay (due to lightweight watermark generation process) in a significant way compared to traditional security techniques with high complexity.
\end{enumerate}
There are two main types of digital watermarking: fragile watermarking and robust watermarking (based on anti-attack properties: fragile watermark becomes undetectable after data being modified, while robust watermarks can survive many forms of distortion)~\cite{cox2007,megias2021}. Robust watermarking is mainly used for copyright protection and is not sensitive to tampering. On the other hand, fragile watermarking is greatly sensitive to altering, and any change in the carrier leads to a failure in the extraction of watermark, that is used in verification of data integrity~\cite{perez2006}. 

Watermarking algorithms consist mainly from three processes: watermark generation, watermark embedding and watermark extraction and verification. Based on fragile watermarking algorithms, the source nodes collect data and generate a watermark. Then, this watermark is embedded in the original data using a predefined rule to construct a watermarked data packet that will be transferred to the destination node through the network. Through the transmission channel, the packet may be subject to many types of attacks and unauthorized access. The destination node receives the packet to extract the digital watermark and separates the original data based on the defined rule used at the source node. The restored data is then used to re-generate a watermark using the generation algorithm applied at the sensing node. Finally, the extracted watermark and the re-generated watermark will be compared to verify data integrity~\cite{Zhang2017}.

In this paper, we focus on securely transmitting captured data and provenance using zero-watermarking approach based on fragile watermarking. The works that are related to the proposed scheme fall into two classes: data integrity using watermarking and provenance security. The notation and their description used in this paper are listed in Table~\ref{tab:table1}. Below, we discuss fragile watermarking-based approaches and provenance security methods.

\subsection{Data Integrity using Digital Watermarking in WSN and IoT Networks}

The literature includes relevant watermarking algorithms used for data integrity and secure transmission in~\gls*{wsn} and~\gls*{iot} environments. Existing watermarking schemes can be classified into two main methods: regular watermarking schemes and zero-watermarking schemes. 

\subsubsection{Regular Watermarking Schemes}

In~\cite{Guo2007}, Guo et al.  propose a fragile digital watermarking algorithm called \gls*{sgw}, which verifies data integrity from the application layer of the data stream. The algorithm groups data based on the key and, then, the hash value of the data from the groups is calculated to be used as a watermark. To save bandwidth, the watermark is then embedded into the \gls*{lsb} of the data found in all groups. In this method, the watermark is used to link all groups and thus detect any deletion of data. Kamel et al. in ~\cite{Kamel2010} optimized Guo et al.'s method by proposing a \gls*{lwc} scheme. In \gls*{lwc}, a dynamic group size is used and the watermark is generated by calculating the hash value of two consecutive groups of data. This leads to less computational overhead, which is an improvement from \gls*{sgw} that calculates the hash value of each data element in the group. The security vulnerabilities of \gls*{sgw} and \gls*{lwc} are addressed in~\cite{Kamel2011}, by proposing a \gls*{fwcd}. In this scheme, the algorithm uses a serial number (SN) that is attached to each group to detect how many group insertions or deletions have occurred in case of any insertion or deletion attack. The watermark is generated using a hash function and a group serial number. All groups are chained with a digital watermark after embedding the watermark of the current group into the previous group in order to bypass any replay attack. To solve the issue of synchronization between sender and receiver nodes in single chaining techniques, a dual-chaining technique is proposed in~\cite{Wang2019}. It generates and embeds fragile watermarks into data using dynamic groups. A reversible watermarking-based algorithm for data integrity authentication is proposed in~\cite{Shi2013}. It applies prediction-error expansion for avoiding any loss in sensory data. Every two adjacent data items are grouped together, and the algorithm uses the first one to generate the watermark, and the other as a carrier for the watermark. Sun et al.~\cite{Sun2013} propose a lossless digital watermarking approach which embeds the generated watermark in the redundant space of data fields. The method does not increase data storage space due to the fixed size of redundant space. However, data integrity is only checked at the base station side. Guoyin et al.~\cite{Zhang2017} propose a watermarking scheme for data integrity based on fragile watermarking in order to solve the problem of resource restrictions in the perception layer of an \gls*{iot} network. They design a \gls*{prw} that calculates the embedding positions for watermarks. The watermarks are generated using the SHA-1 one-way hash function, which is then embedded to the dynamic computed position. This scheme ensures the data integrity at the sink node. It cannot verify integrity along the route of the communication. The drawback of these solutions is that they only provide an end-to-end verification, since watermark generation and verification is based on a group of data packets. Hence, we consider our approach as a better alternative to serve hop by hop data integrity verification between source and final destination of data packets.

Zhou et al.~\cite{Zhou2012} propose a secure data transmission scheme using digital watermarking technique in a \gls*{wsn}. In this scheme, the hash value of two time-adjacent sensitive data is calculated at the source node using a one-way hash function. Then, according to a digital watermarking algorithm, the sensitive data will be embedded into part of the hash sequence as watermark information. The scheme lacks any proof of concept and security analysis to check its robustness against different type of attacks. Another solution for attack detection presented in ~\cite{Alromith2018}. The method develops a \gls*{rwfs} for \gls*{iot} applications by  deploying an en-route filtering that removes injected data at early stage communication based on randomly embedding a watermark in the payload of the packet. The scheme encrypts the data packet before transmission, which encounters additional computation overhead for sensor nodes. In our proposed scheme, only some extracted data features are encrypted to form a sub-watermark that is concatenated to the data packet. To minimize energy consumption, Lalem et al.~\cite{Lalem2016} propose a distributed watermarking technique for data integrity in a \gls*{wsn} using linear interpolation for watermark embedding. The technique allows each node to check the integrity of the received data locally by extracting the watermarked data and generating a new watermark for verification. This method is based on a fixed watermark parameter for all sensor nodes in the network that can be vulnerable to many attacks.

\subsubsection{Zero-watermarking Schemes}

Zero-watermarking is a relatively new digital watermarking method. Each watermarking scheme has a different watermark generation, embedding and extraction process such as unique code (embedded in information hiding schemes), changing position of bits and hash functions  (cryptographic \linebreak schemes) \cite{hameed2018}. In zero-watermarking schemes, watermarks are generated by source node from the extraction of important data features of original data without amendment to the data of these features. Different generation functions can be applied in zero-watermarking. The generated watermarks are not embedded in the data payload, but it is invisibly integrated in the data packet and the data remain unmodified. 

Although several zero-watermarking techniques exist in the related literature, few methods are proposed to protect data integrity in \gls*{iot} environments. In~\cite{Boubiche2015}, a secure data aggregation watermarking-based scheme in homogeneous \gls*{wsn} (SDAW) is presented as a new security technique to protect data aggregation. In this mechanism, watermarks are generated using the \gls*{mac} address of sensor nodes and collected data by a one way hash function (SHA-1). The proposed scheme guarantees secure communication between the aggregation nodes and the base station. However, authors do not provide any security analysis to check the resistance of this scheme against different type of attacks. In our model, we use SHA-2 as a one-way hash function for generating a sub-watermark and provide a detailed security analysis to prove the resistance of our scheme against several attacks. To ensure trustworthiness and data integrity in an \gls*{iot} network, Hameed et al.~\cite{hameed2018} propose a zero-watermarking technique which generates and constructs a watermark in the original data before being transmitted. The generation process of the watermark is based on the original data features (data length, data occurrence frequency and data capturing time). This scheme is shown to be more computationally efficient and requires less energy compared to cryptographic techniques or reversible watermarking schemes. The proposed approach is vulnerable to modification attack, since it only uses data length, data occurrence frequency and data capturing time to generate the watermark. Hence, if data is modified by changing position of data values, the attack will not be detected. 

\subsection{Data Provenance IoT Networks}

The concept of data provenance is used in many fields of study. Each application domain defines provenance in a different way~\cite{Park2008}. In \gls*{iot} networks, we define data provenance as information about the origin of data and how it is routed by forwarding nodes from the source node towards the base station~\cite{Lim2010}. Data provenance guarantees that the data received at the final destination is trusted by the user, verifying that the data is captured by the authorized specific \gls*{iot} node at the stated time and location~\cite{Bertino2018}. Provenance can be represented as a path of nodes from source to destination as shown in Figure~\ref{fig2}. 

Several researchers used the concept of data provenance to identify the origin of data, track the ownership to serve the authenticity of data and assess trustworthiness. Aman et al.~\cite{Aman2017} propose a lightweight protocol for data provenance in \gls*{iot} networks. It uses \gls*{pufs} to identify \gls*{iot} devices and establishing physical security. Wireless links between \gls*{iot} devices and  servers are identified using the wireless channel characteristics. The proposed approach only considers single hop scenarios between \gls*{iot} devices and the server. Kamal et al.~\cite{Kamal2018} present a lightweight protocol for a multi-hop \gls*{iot} network to provide security for data and achieve data provenance. The protocol uses link fingerprints generated from the \gls*{rssi} of \gls*{iot} nodes in the network. Data provenance is achieved by attaching the encoded link fingerprint to the header of data packet at each hop. The packet header is then decoded in sequence at the server. Ragib et al.~\cite{Hasan2009} propose a provenance aware system that ensures integrity and confidentiality. The scheme implements provenance tracking of data at the application layer through encryption and incremental chained signature. Similarly, Syalim et al.~\cite{syalim2010} extend the same method to a directed acyclic graph (DAG) model of provenance database. The approach preserves integrity and confidentiality of nodes by using digital signatures and double encryption. These solutions do not consider the characteristics of sensor networks. Provenance information grows very fast, which requires transmitting large amount of provenance information with data packets (\ie increasing the bandwidth overhead).  

Sultana et al.~\cite{Sultana2011} establish a data provenance mechanism to detect malicious packet dropping attacks. The method relies on the inter-packet timing characteristics after embedding provenance information. It detects the packet loss based on the distribution of the inter-packet delays and then identifies the presence of an attack to finally localize the malicious node or link. Lim et al.~\cite{Lim2010} assess the trustworthiness of data items in a sensor network. They use data provenance and their values to compute trust scores for data items and nodes in the network. These scores are used to provide the level of trustworthiness of a data item as well as nodes. These schemes provide security and trustworthiness for sensor networks, but do not consider the retrieval of data provenance against different types of attacks. In this paper, we propose a strategy to securely transmit provenance using a tamper-proof database. In our approach, the IP address is used as the identity of a node and thus encoded in the provenance. At the gateway, after data integrity verification, the set of watermarks stored in the database is queried and the list of nodes is extracted, so that the gateway can easily and securely construct the provenance (\ie data path).

\begin{figure*}[!htbp]
\begin{center}
    \includegraphics[scale=0.46]{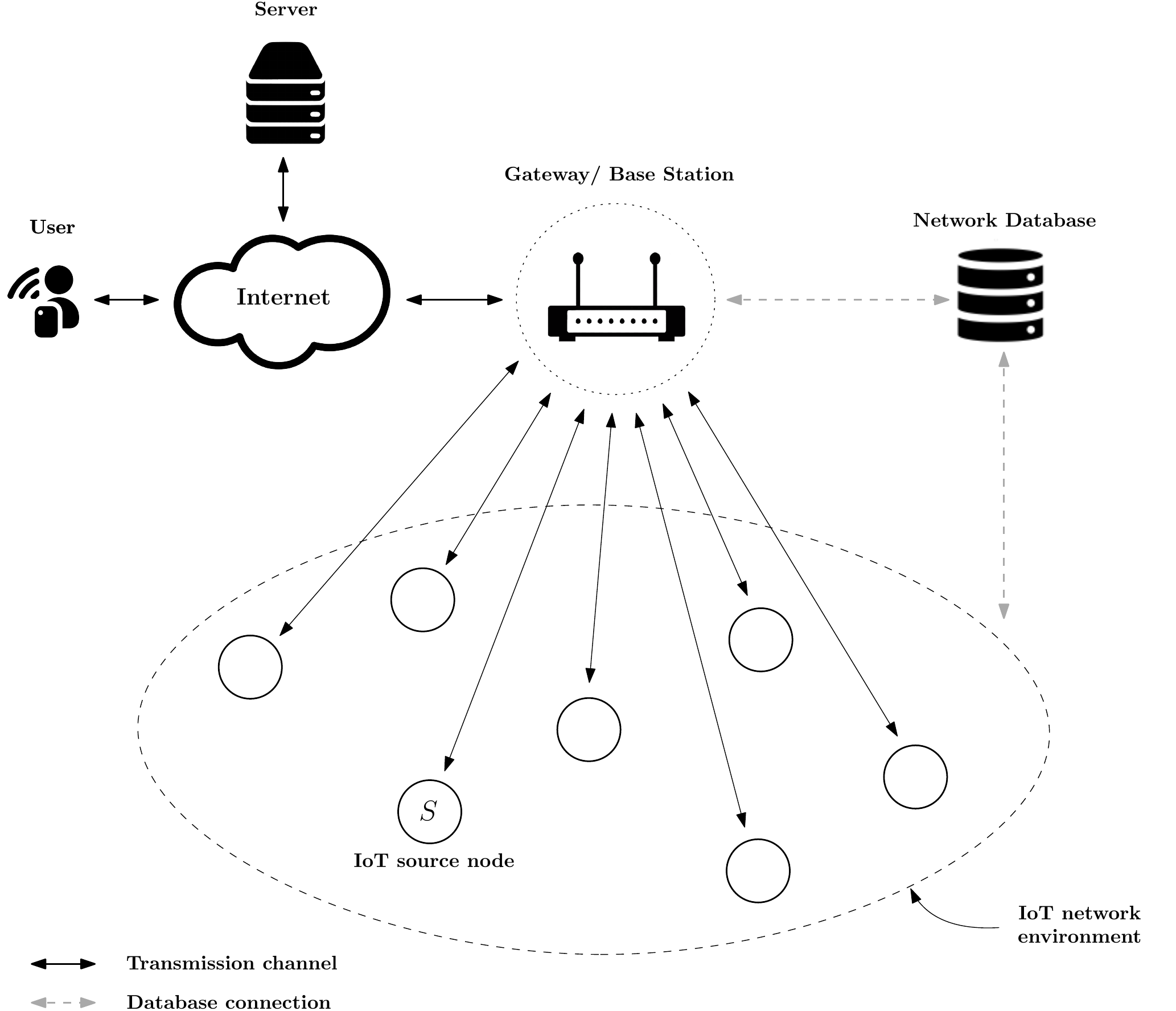}
\end{center}
\caption{Single hop network model.}
\label{fig1}
\end{figure*}

%%%%%%%%%%%%%%%%%%%%%% Table of Notations %%%%%%%%%%%%%%%%%%%%%%%

\begin{table*}[!htbp]
\centering
\caption{System Notation and Parameters}
\begin{tabular}{ p{2cm}|p{10cm}  }
\hline
\thead{Notation} & \thead{Description} \\
\hline

$d_{n,k}$ & Captured data packet $k$ from node $n$ \\
$N$ & Number of nodes in the network \\
$H$ & Number of hops the data packet $d_{n,k}$ routed through \\
$w_{ip}$ &  \gls*{iot} Device $n$ IP Address \\
$w_{t}$ & Sensed data ($d_n$) capturing time \\ 
$sw_{f_{n,k,i}}$ & Sub-watermark $i$ generated using data features of $d_{n,k}$\\
$sw_{h_{n,k,i}}$ & Sub-watermark $i$ generated using hash function for $d_{n,k}$\\
$E(sw_{f_{n,k,i}})$ & Encrypted sub-watermark $sw_{f_{n,k,i}}$\\
$W_{F_{n,k,i}}$ & Final generated watermark $i$ of data packet $k$ from node $n$\\
$d_{(n,k)W_{F_{n,k,i}}}$ & Watermarked data \\
$R(W_{F_{n,k,i}}$ & Re-generated watermark \\
$R(sw_{f_{n,k}})$ & Re-generated sub-watermark from data features\\
$R(sw_{h_{n,k,i}})$ & Re-generated sub-watermark from hash function\\
$p_{n,k,i}$ & provenance record $i$ of data packet $k$ from node $n$\\
$P_{n,k}$ & set of provenance records of data packet $k$ from node $n$\\
$||$ & Concatenation \\
$ENC()$ & Encryption function \\
$DEC()$ & Decryption function \\
$H()$ & One-way hash function \\
$K_j$ & $j^{th}$ generated and distributed secret key for encryption/decryption \\
$QRY()$ & Query function from network database \\
$STR()$ & Store function to network database \\
$I_l$ & Intermediate node $l$ \\
$S_n$ & Source node $n$ \\
$G$ & Base station or Gateway \\
$q$ & Number of bits to be deleted by an attacker \\

\hline
\end{tabular}
\label{tab:table1}
\end{table*}

\section{Proposed Construction}
\label{sec:sysmodel}
In this paper, we propose a zero-watermarking approach to verify data integrity at each hop of an \gls*{iot} environment (\ie from source node to gateway). Additionally, we ensure secure provenance of sensory data using a tamper-proof database connected to the gateway and to each node in the network. Data provenance information is used in the watermarking generation process and embedded with captured data packets to be used in the data integrity and secure provenance verification. The system is composed of the following entities:
\begin{itemize}
    \item \textbf{\gls*{iot} Source Node:} a small sensing device that collects data from surrounding environment. At each node, watermarking generation and embedding processes are applied to each data packet captured. The device performs some operations and communication procedures in the network.
    \item \textbf{\gls*{iot} Intermediate Node:} an advanced sensing device with more power and computation capabilities, responsible for forwarding data packets from source nodes to the base station. It also performs watermark generation, embedding and storing on the received data packets. Data integrity will be checked for each data packet forwarded at this stage.
    \item \textbf{Base Station or gateway:} receive the forwarded data packets for data processing. Checks data integrity and provenance recovery, and applies the attack detection procedure.
    \item \textbf{Network Database:} a secure network database which is connected to the network nodes and gateway. It stores provenance information at each hop of the data path. 
\end{itemize}

\subsection{Zero-Watermarking for Data Integrity and Provenance Network Model}
In the proposed network model, the network is assumed to consist of $N$ \gls*{iot} devices that are distributed in an \gls*{iot} network. The network is deployed in an $L \times W$ area. Devices are connected to a gateway or base-station that is the management and controller unit. Nodes are of two types: normal sensor nodes and intermediate sensor nodes. Sensory data is routed from normal source nodes to the gateway through intermediate nodes. This implies that intermediate nodes have $m$ times more energy than normal nodes (\ie energy of a normal node = $E_0$, energy of intermediate node = $E_0 + m \times E_0$). Furthermore, a tamper-proof database is connected to the gateway and to each node. It is assumed that the database cannot be compromised by the attacker. The model consists of two main scenarios: single-hop and multi-hop. In the single-hop scenario, \gls*{iot} devices transmits sensory data directly to the gateway through the transmission channel, as shown in Figure~\ref{fig1}. However, in the multi-hop scenario, sensory data is routed from the source node to the gateway through intermediate nodes as shown in Figure~\ref{fig2}.

\begin{figure*}[!htbp]
\begin{center}
    \includegraphics[scale=0.43]{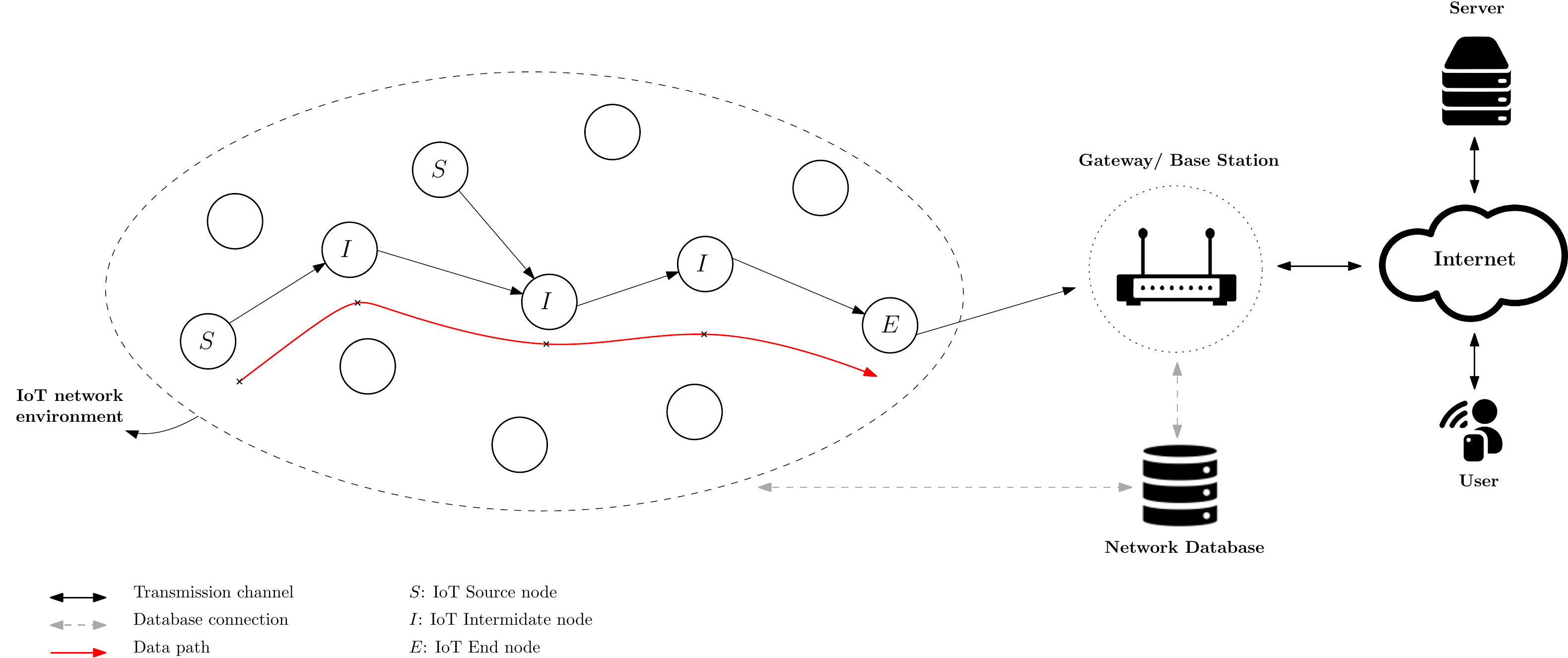}
\end{center}
\caption{Multi hop network model.}
\label{fig2}
\end{figure*}

\subsubsection{Single Hop Model}
In this scenario, the process of data integrity and secure provenance is composed of different units that form the overall system model as described in Figure~\ref{fig3}. Sensor nodes capture data from the surrounding environment and send it to the feature extraction unit. The IP address of the \gls*{iot} device and data capturing time are extracted and encrypted to generate a sub-watermark in the first sub-watermark generation unit. This sub-watermark forms the provenance record of a particular data packet. A hash function is used to generate another sub-watermark that is concatenated with the first one to generate a final watermark. The generated final watermark is then stored in a tamper-proof database. The data packet is then sent to the gateway through the transmission channel. At the gateway, data is received and forwarded to the zero-watermark re-generation unit. After the re-generation process, the stored watermark is queried from the database to be compared with the re-generated watermark in the watermark verification unit for provenance integrity check. A double verification procedure is applied for both integrity and provenance. At this stage, the gateway detects whether data and provenance is altered or not and performs either attack procedure or validates the origin of data received. 

\begin{figure*}[t]
\centerline{\includegraphics[scale=0.45]{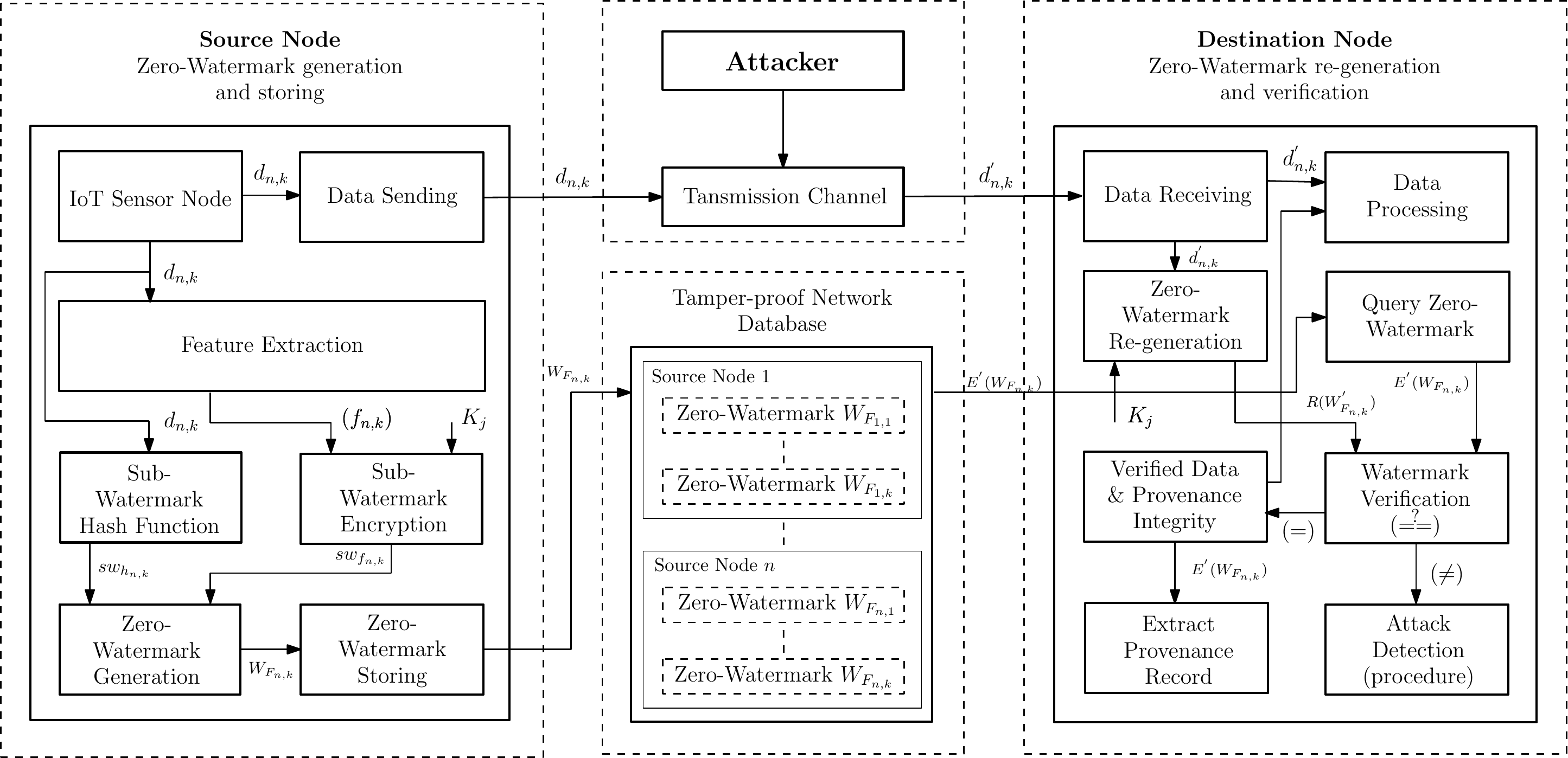}}
\caption{Zero-watermark generation, storing and verification block diagram in single hop scenario.}
\label{fig3}
\end{figure*}

\subsubsection{Multi Hop Model}
The watermark generation, embedding, extraction and verification processes of the multi-hop scenario are shown in Figure~\ref{fig4}. In this model, the data capturing time and IP address are extracted from sensed data packets to generate a sub-watermark, which is then encrypted using a secret key and concatenated with a generated hash value of data payload to construct the final watermark. The first sub-watermark or provenance record is stored in the network database and the final watermark is concatenated with the sensed data $d_{n,k}$ to be forwarded to the next intermediate node through the transmission channel. At the next hop node, the watermarked data is received. The watermark is then extracted from the received data packet. The received data is then used to re-generate a new sub-watermark that will be forwarded to the verification unit along with the extracted watermark and a queried provenance record. The intermediate node takes a decision whether an attack is detected or not. Based on this decision, it performs an attack detection procedure or generates the next-hop watermark that undergoes the same procedure of the source node (generation, embedding and storing); it uses new extracted features and provenance information. The watermarked data reaches the final destination (\ie gateway) through transmission channel. The last embedded watermark is separated from the watermarked data and a final sub-watermark is re-generated. The data integrity unit accepts extracted watermark, re-generated sub-watermark and queried provenance record as input values to check whether data or provenance is modified or not. After that, the gateway performs two procedures based on the verification result: attack detection procedure or provenance validation. 

\begin{figure*}[t]
\centerline{\includegraphics[scale=0.33]{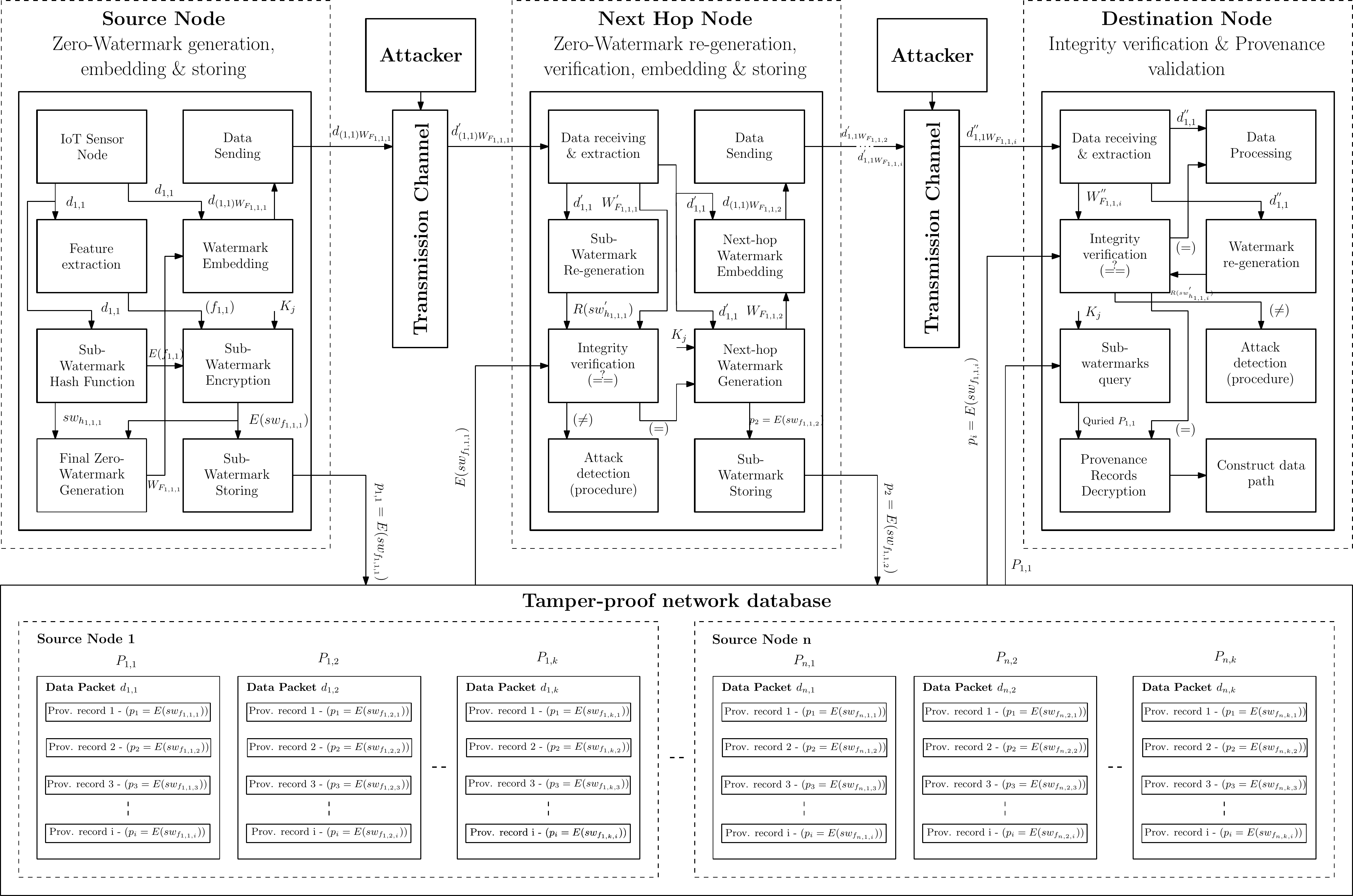}}
\caption{Zero-watermark generation, storing and verification block diagram in multi-hop scenario.}
\label{fig4}
\end{figure*}

\subsection{Security Assumptions} \label{sec:assumptions}

We consider a set of security assumptions for our proposed system as follows: 

\begin{itemize}
    \item Nodes in the network are not trusted entities, \ie, these nodes may be malicious. The protocols provided to guarantee the secure transfer of provenance information that is applied in intermediate nodes and gateway are proven to work properly in the presence of malicious nodes.
    \item The network database is a trusted and secure entity, it cannot be compromised by an external attacker to access or use its content.
    \item To allow only authorized gateways to access the network database and query provenance information, only registered gateways are authorized and the query is applied after checking the integrity of the data received from the last hop of the data path. 
    \item The database temporarily stores provenance information of a data packet at each hop from source to destination, only after being proved as trusted data. This linage can be retrieved once (by authorized gateway) and then it is removed from the database.
    \item The hashing functions used in the system are secure
    and cannot be inverted.
    \item The communication of extracted data features and provenance information (sub-watermarks) between source nodes and intermediate nodes, and between intermediate nodes and gateway, must be secure. Provenance information is encrypted using symmetric cryptography and selected data (for integrity check) is hashed using a one-way hash function.
    \item Symmetric cryptography is restricted to the encryption of short binary strings forming the extracted data feature sub-watermark. Source node, intermediate node and gateway share secret-keys to be used in different steps of the algorithms (encryption/decryption).

    \item Secret keys are changed and redistributed after a short random number of watermark generation processes. 
    \item The zero-watermarking method used to embed provenance information is transparent, fragile and secure enough for \gls*{iot} network applications.
\end{itemize}

\subsection{Threat Model} \label{sec:threat}

There is a number of different attacks that may be applied against the proposed system. Attack models to deceive and perform security breaching on different network entities require another party to obtain secret information or access network database. A threat model similar to the one used in~\cite{Cui2018} is applied in our scheme. The attacker can perform two types of attacks: passive and active attacks (\eg external attacks).
\begin{enumerate}
    \item \textbf{Passive attack:} An attacker observes secret information by passively eavesdropping data. The attacker performs an eavesdropping attack that aims to obtain data information through listening to data transmission line in the wireless medium.
    \item \textbf{Active or External attack:} A malicious attack aiming to destroy information by modifying data packets through launching different kinds of operations. An external or active adversary can launch three main attacks: 
    \begin{enumerate}
        \item \textbf{Replay attack:} Data packets are captured by an adversary and then resent in the future at a different time interval to deceive intermediate nodes or the gateway.
        \item \textbf{Integrity attack:} An attacker inserts false value(s) into the data packet at the transmission channel to deceive the gateway. Also, the attacker may delete elements of the data packet. 
        \item \textbf{Modification attack:} In this attack, data is modified by an attacker without knowledge of the data content.
        \item \textbf{Database authentication attack:} An attacker aims to detect and retrieve provenance information stored in the network database.
    \end{enumerate}
\end{enumerate}

\subsection{Proposed Algorithms}
In this section, a precise algorithmic presentation of ZIRCON to conduct the zero-watermarking scheme is described in details. 

    \begin{algorithm}[t]
    \caption{Watermark generation and storing}\label{alg1}
    %\textbf{input:} $d_{n,k}$ 
    %\textbf{output:} $W_{F_{n,k}$
    \begin{algorithmic}[1]
    \Procedure{Watermark generation and storing}{}
    \State $w_{ip} \gets $ \gls*{iot} Device $n$ IP Address
    \State $w_{t} \gets $ captured data.sensing time ($d_{n,k}$)
    \State $sw_{f_{n,k}} \gets w_{ip} $ $||$ $w_{t}$
    \State $p_{n,k}$ $\gets$ $E(sw_{f_{n,k}}) \gets$ ENC($K_j$,$sw_{f_{n,k}}$) 
    \State $sw_{h_{n,k}} \gets H(d_{n,k})$ \Comment{Select first 8 bytes of hash output}
    \State $W_{F_{n,k}} \gets E(sw_{f_{n,k}})$ $ || $ $sw_{h_{n,k}}$
    \State $STR(W_{F_{n,k}})$
    \State $Send(d_{n,k})$
    \EndProcedure
    \end{algorithmic}
    \end{algorithm}

\subsubsection{Single hop scenario.}
In this scenario, two algorithms are proposed: watermark generation and storing, and watermark verification. 

\begin{enumerate}
\setlength{\itemsep}{10pt}

    \item \textbf{Watermark generation and storing:} Algorithm~\ref{alg1} shows the process of generating and storing a watermark in a single-hop scenario. It accepts sensed data from the \gls*{iot} device to produce a final watermark. The algorithm extracts the IP address and the data capturing time from the source node and combines it to generate a sub-watermark $sw_{f_{n,k}}$ as shown in Lines 2-4 of Algorithm~\ref{alg1}. The sub-watermark is then encrypted using the secret key $K_j$ to obtain a provenance record $p_{n,k} = E(sw_{f_{n,k}},K_j)$. Another sub-watermark $sw_{h_{n,k}}$ is generated from the hash value of data payload using a one-way hash function. These two generated sub-watermarks are concatenated to form a final watermark $W_{F_{n,k}}$ as:
    
    \begin{equation} \label{eq:1}
    W_{F_{n,k}} = E(w_{ip} \, || \, w_{t} \,,K_j) \, || \, H(d_{n,k})
    = E(sw_{f_{n,k}},K_j) \, || \, sw_{h_{n,k}}
    \end{equation}
    
    Here, $||$ is the concatenation operator, $W_{F_{n,k}}$ $( 1 \leq n \leq N,$ where $N$ is the number of nodes in the network$)$ is the final watermark, $H$ is a lightweight, secure hash function, and $n$ is the particular node number. The watermark generation algorithm uses the SHA-2 hash function to calculate the hash value. The advantage of using the SHA-2 hash function over other hash algorithms is that SHA-2 has a lightweight feature that uses 65\% less memory than other algorithms, such as the MD5 hash function (which has several vulnerability issues), which is needed in resource-constrained networks~\cite{Ganesan2003}. After the generation procedure, the final watermark is stored in the network database and the data packet is sent to the base station as shown in Lines 7-9.

    \begin{algorithm}[!htbp]
    \caption{Watermark verification}\label{alg2}
    \hspace*{\algorithmicindent} \textbf{input:} $d^{'}_{n,k}$ \\
    \hspace*{\algorithmicindent} \textbf{output:} verified/not verified
    \begin{algorithmic}[1]
    \Procedure{Watermark querying and verification}{}
    \State $Receive(d^{'}_{n,k})$
    \State $R(W^{'}_{n,k}) \gets $ REDO Algorithm 1 \Comment{Re-generate watermark}
    \State $R(sw^{'}_{h_{n,k}}) \gets $ $EXTRACT(R(W^{'}_{n,k}))$
    \State $W_{F_{n,k}} \gets $ $QRY(W_{F_{n,k}})$ \Comment{Query zero-watermark from database}
    \State $sw_{h_{n,k}} \gets $ $EXTRACT(W_{F_{n,k}})$
    \State \textbf{if} ($R(sw^{'}_{h_{n,k}}) = sw_{h_{n,k}}$) \textbf{then}
    \State \hspace{10pt} Data Integrity Verified 
    \State \hspace{10pt} $E(sw_{f_{n,k}}) \gets $ $EXTRACT(W_{F_{n,k}})$ 
    \State \hspace{10pt} $E(R(sw^{'}_{f_{n,k}})) \gets $ $EXTRACT(R(W^{'}_{n,k}))$
    \State \hspace{10pt} \textbf{if} ($E(R(sw^{'}_{f_{n,k}})) = E(sw_{f_{n,k}})$) \textbf{then}
    \State \hspace{20pt} Provenance Verified 
    \State \hspace{20pt} $p_{n,k} \gets $ $DEC(E(sw_{f_{n,k}}),K_j)$  \Comment{Sub-watermark decryption}
    \State \hspace{20pt} Extract \gls*{iot} device ($n$) IP address
    \State \hspace{20pt} Check origin of data packet ($d^{'}_{n,k}$)
    \State \hspace{20pt} Process data ($d^{'}_{n,k}$)
    \State \hspace{10pt} \textbf{else} 
    \State \hspace{20pt} Provenance Not Verified
    \State \hspace{20pt} Discard data ($d^{'}_{n,k}$)
    \State \hspace{20pt} Perform attack procedure
    \State \hspace{20pt} Delete $W_{F_{n,k}}$ from network database
    \State \hspace{10pt} \textbf{end if}
    \State \textbf{else} 
    \State \hspace{10pt} Integrity Not Verified
    \State \hspace{10pt} Discard data ($d^{'}_{n,k}$)
    \State \hspace{10pt} Perform attack procedure
    \State \hspace{10pt} Delete $W_{F_{n,k}}$ from network database
    \State \textbf{end if}
    \EndProcedure
    \end{algorithmic}
    \end{algorithm}
    
    \item \textbf{Watermark verification:} The process of verifying data integrity and validating data provenance in a single-hop scenario is shown in Algorithm~\ref{alg2}, which takes the received data $d^{'}_{n,k}$ as an input. Then, a re-generation procedure is performed to re-generate the watermark $R(W^{'}_{n,k})$ and the stored watermark $W_{F_{n,k}}$ is queried from the database. A comparison operation is then applied on the re-generated sub-watermark $R(sw^{'}_{h_{n,k}})$ and the queried sub-watermark $sw_{h_{n,k}}$. If the sub-watermarks are the same, data integrity is verified. Then, another comparison operation is performed that compares the re-generated sub-watermark $E(R(sw^{'}_{f_{n,k}}))$ and the second queried sub-watermark $E(sw_{f_{n,k}})$, for provenance integrity check, as shown in Line 11. If these two sub-watermarks are the same, provenance integrity is verified and the provenance record $p_{n,k}$ that contains provenance information is decrypted using the secret key $K_j$. The IP address and data capturing time are obtained from $p_{n,k}$. Provenance is then validated and data is ready for processing. After that, the stored provenance record $p_{n,k}$ of received data packet $d^{'}_{n,k}$ may be deleted from the database, after being used for security analysis. Whereas, if sub-watermarks were not the same, data $d^{'}_{n,k}$ will be discarded and an attack procedure is performed (check the type of attack or origin of data is being altered). The stored provenance $p_{n,k}$ of the discarded data, after attack detection, will be deleted from the database. 
    
\end{enumerate}

    \begin{algorithm}[!htbp]
    \caption{Watermark generation and embedding}\label{alg3}
    \hspace*{\algorithmicindent} \textbf{input:} $d_{n,k}$ \\
    \hspace*{\algorithmicindent} \textbf{output:} $d_{(n,k)W_{F_{n,k,i}}}$
    \begin{algorithmic}[1]
    \Procedure{Watermark generation and storing}{}
    \State $w_{ip} \gets $ \gls*{iot} Device ($n$) IP Address
    \State $w_{t} \gets $ captured data.sensing time ($d_{n,k}$)
    \State $sw_{f_{n,k,i}} \gets w_{ip} $ $||$ $w_{t}$
    \State $p_{n,k,i} \gets$ $E(sw_{f_{n,k,i}}) \gets$ $ENC(sw_{f_{n,k,i}},K_j)$
    \State $sw_{h_{n,k,i}} \gets H(d_{n,k})$ \Comment{Select first 8 bytes of hash output}
    \State $W_{F_{n,k,i}} \gets p_{n,k,i}$ $ || $ $sw_{h_{n,k,i}}$
    \State $STR(p_{n,k,i})$
    \EndProcedure
    \Procedure{Watermark Embedding}{}
    \State $d_{(n,k)W_{F_{n,k,i}}}$ $\gets$ $d_{n,k}$ $||$ $W_{F_{n,k,i}}$
    \State $Send(d_{(n,k)W_{F_{n,k,i}}})$
    \EndProcedure
    \end{algorithmic}
    \end{algorithm}

    \begin{algorithm}[!htbp]
    \caption{Watermark verification and re-embedding}\label{alg4}
    \hspace*{\algorithmicindent} \textbf{input:} $d^{'}_{(n,k)W_{F_{n,k,i}}}$ \\
    \hspace*{\algorithmicindent} \textbf{output:} verified/not verified, $W_{F_{n,k,i}}$, $d^{'}_{(n,k)W_{F_{n,k,i}}}$
    \begin{algorithmic}[1]
    \Procedure{Watermark verification and re-embedding}{}
    \State $Receive(d^{'}_{(n,k)W_{F_{n,k,i}}})$
    \State Extract Watermarked Data into $d^{'}_{n,k}$ and $W^{'}_{F_{n,k,i}}$
    \State $R(sw_{h_{n,k,i}}) \gets $ REDO Algorithm 1 \Comment{Re-generate sub-watermark}
    \State $sw^{'}_{h_{n,k,i}} \gets $ $EXTRACT(W^{'}_{F_{n,k,i}})$
    \State \textbf{if} ($R(sw_{h_{n,k,i}}) = sw^{'}_{h_{n,k,i}}$) \textbf{then}
    \State \hspace{10pt} Integrity Verified
    \State \hspace{10pt} $E(sw_{f_{n,k,i}}) \gets $ $QRY(p_{n,k,i})$ \Comment{Query provenance record}
    \State \hspace{10pt} $E(sw^{'}_{f_{n,k,i}}) \gets $ $EXTRACT(W^{'}_{F_{n,k,i}})$
    \State \hspace{10pt} \textbf{if} ($E(sw^{'}_{f_{n,k,i}}) = E(sw_{f_{n,k,i}})$) \textbf{then}
    \State \hspace{20pt} Provenance Integrity Verified 
    \State \hspace{20pt} \textbf{Generate next hop watermark}\{
    \State \hspace{30pt} $w_{ip} \gets $ next hop device IP Address
    \State \hspace{30pt} $w_{t} \gets $ time watermarked data packet ($d^{'}_{(n,k)W_{F_{n,k,i}}}$) received
    \State \hspace{30pt} $i++$ \Comment{update next hop watermark index}
    \State \hspace{30pt} $sw_{f_{n,k,i}} \gets w_{ip} $ $||$ $w_{t}$ 
    \State \hspace{30pt} $p_{n,k,i} \gets$ $E(sw_{f_{n,k,i}}) \gets$ $ENC(sw_{f_{n,k,i}},K_j)$
    \State \hspace{30pt} $sw_{h_{n,k,i}} \gets $ Hash value from $R(sw_{h_{n,k,i}})$
    \State \hspace{30pt} $W_{F_{n,k,i}} \gets E(sw_{f_{n,k,i}})$ $ || $ $sw_{h_{n,k,i}}$\}
    \State \hspace{20pt} $d^{'}_{(n,k)W_{F_{n,k,i}}} $ $\gets$ $d^{'}_{n,k}$ $||$ $W_{F_{n,k,i}}$ \Comment{Embedding next hop watermark}
    \State \hspace{20pt} $STR(p_{n,k,i})$ \Comment{Sub-watermark storing} 
    \State \hspace{20pt} $Send(d^{'}_{(n,k)W_{F_{n,k,i}}})$
    \State \hspace{10pt} \textbf{else} 
    \State \hspace{20pt} Provenance no verified/attack detection
    \State \hspace{20pt} Discard data $d^{'}_{n,k}$ \& Perform attack procedure
    \State \hspace{20pt} Delete $P_{n,k}$ from network database
    \State \hspace{10pt} \textbf{end if}
    \State \textbf{else} 
    \State \hspace{10pt} Not verified/attack detection
    \State \hspace{10pt} Discard data $d^{'}_{n,k}$
    \State \hspace{10pt} Perform attack procedure
    \State \hspace{10pt} Delete $P_{n,k}$ from network database
    \State \textbf{end if}
    \EndProcedure
    \end{algorithmic}
    \end{algorithm}

\subsubsection{Multi hop scenario.}
Three algorithms are proposed in this scenario: watermark generation and embedding, watermark verification and re-embedding, and data integrity verification and provenance reconstruction.

\begin{enumerate}
\setlength{\itemsep}{10pt}

    \item \textbf{Watermark generation and embedding} Algorithm~\ref{alg3} describes the working process of two procedures: watermark generation and storing, and watermark embedding in the multi-hop scenario. The algorithm accepts the captured data $d_{n,k}$ as an input obtained from the source node that is sensing data from the surrounding environment. In the first procedure, two inputs are used for generating the first sub-watermark such as the \gls*{iot} device IP address $w_{ip}$ and the data sensing time $w_{t}$. The sub-watermark $sw_{f_{n,k,i}} $ is formed by appending these values. To secure the provenance information, $sw_{f_{n,k,i}}$ is encrypted using secret key $K_j$. The encrypted value forms the provenance record $p_{n,k,i}$. Another sub-watermark $sw_{h_{n,k,i}}$ is generated from the hash value of the data payload. Finally, the final watermark $W_{F_{n,k,i}}$ is produced by concatenating the two sub-watermarks $sw_{f_{n,k,i}}$ and $sw_{h_{n,k,i}}$ as in Equation~\eqref{eq:2}. Provenance record $p_{n,k,i}$ is then stored in the network database as shown in Line 8. In the second procedure, the watermarked data $d_{(n,k)W_{F_{n,k,i}}}$ is produced by concatenating the final watermark $W_{F_{n,k,i}}$ with the captured data packet $d_{n,k}$ as shown in Equation~\eqref{eq:3}.
    
    \begin{equation} \label{eq:2}
    W_{F_{n,k,i}} = E(w_{ip} \, || \, w_{t} \,,K_j) \, || \, H(d_{n,k})
    = E(sw_{f_{n,k,i}},K_j) \, || \, sw_{h_{n,k,i}}
    \end{equation}
    
    \begin{equation} \label{eq:3}
    d_{(n,k)W_{F_{n,k,i}}} = d_{n,k} \, || \, W_{F_{n,k,i}}
    \end{equation}

    \begin{algorithm}[!htbp]
    \caption{Data integrity verification and provenance reconstruction}\label{alg5}
    \hspace*{\algorithmicindent} \textbf{input:} $d^{''}_{(n,k)W_{F_{n,k,i}}}$ \\
    \hspace*{\algorithmicindent} \textbf{output:} verified/not verified, provenance construction
    \begin{algorithmic}[1]
    \Procedure{Watermark restoring and verification}{}
    \State $Receive(d^{''}_{(n,k)W_{F_{n,k,i}}})$
    \State Extract Watermarked Data into $d^{''}_{n,k}$ and $W^{''}_{F_{n,k,i}}$
    \State $R(sw_{h_{n,k,i}}) \gets $ REDO Algorithm 3 \Comment{Only sub-watermark generation procedure}
    \State $sw^{''}_{h_{n,k,i}} \gets $ $EXTRACT(W^{''}_{F_{n,k,i}})$
    \State \textbf{if} ($R(sw_{h_{n,k,i}})  =  sw^{''}_{h_{n,k,i}}$) \textbf{then}
    \State \hspace{10pt} Data integrity verified
    \State \hspace{10pt} $E(sw_{f_{n,k,i}}) \gets $ $QRY(p_{n,k,i})$ \Comment{Query last provenance record}
    \State \hspace{10pt} $E(sw^{''}_{f_{n,k,i}}) \gets $ $EXTRACT(W^{''}_{F_{n,k,i}})$
    \State \hspace{10pt} \textbf{if} ($E(sw_{f_{n,k,i}})  =  E(sw^{''}_{f_{n,k,i}})$) \textbf{then}
    \State \hspace{20pt} Provenance integrity verified
    \State \hspace{20pt} $P_{n,k} \gets QRY(P_{n,k})$  \Comment{Query stored set of provenance records}
    \State \hspace{20pt} \textbf{for} (index $i$, $1 \le i \le H$, $i++$) \textbf{do}
    \State \hspace{30pt} Extract $E(sw_{f_{n,k,i}})$ of each $p_{n,k,i}$ from $P_{n,k}$
    \State \hspace{30pt} $sw_{f_{n,k,i}} = DEC(E(sw_{f_{n,k,i}}),K_j)$ 
    \State \hspace{30pt} Extract provenance information from each sub-watermark
    \State \hspace{20pt} \textbf{end for}
    \State \hspace{20pt} Construct data path of $d^{''}_{n,k}$
    \State \hspace{10pt} \textbf{else} 
    \State \hspace{20pt} Provenance integrity is not verified/ attack detected
    \State \hspace{20pt} Discard data $d^{''}_{n,k}$
    \State \hspace{20pt} Perform attack procedure
    \State \hspace{20pt} Delete $P_{n,k}$ from network database 
    \State \hspace{10pt} \textbf{end if}
    \State \textbf{else} 
    \State \hspace{10pt} Data integrity not verified/ attack detected
    \State \hspace{10pt} Discard data $d^{''}_{n,k}$
    \State \hspace{10pt} Perform attack procedure
    \State \hspace{10pt} Delete $P_{n,k}$ from network database
    \State \textbf{end if}
    \EndProcedure
    \end{algorithmic}
    \end{algorithm}

    \item \textbf{Watermark verification and re-embedding:} At the next hop, a watermark verification and re-embedding algorithm is applied as shown in Algorithm~\ref{alg4}. To verify data integrity at the next node, the algorithm accepts the watermarked data $d^{'}_{(n,k)W_{F_{n,k,i}}}$ as an input. The captured data $d^{'}_{n,k}$ and watermark $W^{'}_{F_{n,k,i}}$ are extracted from $d^{'}_{(n,k)W_{F_{n,k,i}}}$. A new sub-watermark $R(sw_{h_{n,k,i}})$ is re-generated from $d^{'}_{n,k}$ by using the first procedure of Algorithm~\ref{alg3} and $sw^{'}_{h_{n,k,i}}$ is extracted from $W^{'}_{F_{n,k,i}}$. Then a comparison operation is applied on the sub-watermark values of $R(sw_{h_{n,k,i}})$ and $sw^{'}_{h_{n,k,i}}$ to check whether data is altered or not. If data integrity is verified, $E(sw_{f_{n,k,i}})$ is obtained from querying the provenance record $p_{n,k,i}$ from the network database. Another sub-watermark $E(sw^{'}_{f_{n,k,i}})$ is extracted from $W^{'}_{F_{n,k,i}}$ for provenance validation. Then, a comparison operation is applied on $E(sw_{f_{n,k,i}})$ and $E(sw^{'}_{f_{n,k,i}})$. If both sub-watermarks are the same, provenance integrity is verified and a new watermark is generated using the same procedure of Algorithm~\ref{alg3} as shown in Lines 12-19. The new generated watermark $W_{F_{n,k,i}}$ is formed of the next hop node IP address and the watermarked data packet receiving time, and the same hash value of the data packet obtained from the re-generated sub-watermark $R(sw_{h_{n,k,i}})$ using Equation~\eqref{eq:2}. The watermark $W_{F_{n,k,i}}$ is concatenated with data $d^{'}_{n,k}$ to form a watermarked data packet as shown using Equation~\eqref{eq:3}. Then, the new generated sub-watermark $E(sw_{f_{n,k,i}})$ or provenance record $p_{n,k,i}$ is stored in the network database as shown in Line 20. However, if $E(sw_{f_{n,k,i}})$ and $E(sw^{'}_{f_{n,k,i}})$ are not the same, the provenance is not verified and the data is discarded. Also, $P_{n,k}$ of received watermarked data packet $d^{'}_{(n,k)W_{F_{n,k,i}}}$ is deleted from the database and an attack procedure is applied. If data integrity is not verified, data will be also discarded and an attack procedure will be applied. Also, all stored provenance records of $P_{n,k}$ related to this data packet will be deleted from the database.

    \item \textbf{Data integrity verification and provenance reconstruction:} The process of data integrity verification and provenance reconstruction at the gateway is shown in Algorithm~\ref{alg5}. The verification procedure relies on four main conditions:
    
    \vspace{0.1cm}
    
    \begin{enumerate}
    \setlength{\itemsep}{5pt}
        \item The origin of data packet based on the source IP $address$.
        \item The freshness of the timestamp $w_{t}$ included in the watermark.
        \item The hop by hop integrity and provenance validation.
        \item Verifying the data measurement using the hash value.
    \end{enumerate}
    
    \vspace{0.1cm}
    
    The received watermarked data $d^{''}_{(n,k)W_{F_{n,k,i}}}$ is extracted into $d^{''}_{n,k}$ and $W^{''}_{F_{n,k,i}}$. The gateway re-generates the sub-watermark $R(sw_{h_{n,k,i}})$ by performing the generation process of Algorithm~\ref{alg3} as shown in Line 4 and $sw^{''}_{h_{n,k,i}}$ is extracted from $W^{''}_{F_{n,k,i}}$. The extracted sub-watermark $sw^{''}_{h_{n,k,i}}$ and the re-generated sub-watermark $R(sw_{h_{n,k,i}})$ will be compared using a comparison operation to check data integrity. If data is not altered, the gateway queries the last provenance record $p_{n,k,i}$ of the received watermarked data packet $d^{''}_{(n,k)W_{F_{n,k,i}}}$ from the database. Then, $E(sw^{''}_{f_{n,k,i}})$ is extracted from $W^{''}_{F_{n,k,i}}$. The gateway performs a comparison operation for $E(sw^{''}_{f_{n,k,i}})$ and $E(sw_{f_{n,k,i}})$ (i.e. last stored provenance record). If both values are the same, provenance is verified, the gateway queries the set of stored provenance records $P_{n,k}$ from the database and extracts the encrypted sub-watermarks $E(sw_{f_{n,k,i}})$ of each $p_{n,k,i}$. At Line 15, the secret key $K_j$ is used to decrypt $E(sw_{f_{n,k,i}})$ and obtain the sub-watermarks $sw_{f_{n,k,i}}$ containing provenance information of the received data packet. The gateway constructs the data path from provenance information obtained. If data integrity or provenance is not verified, data will be discarded and an attack procedure is performed and $P_{n,k}$ of received watermarked data packet $d^{''}_{(n,k)W_{F_{n,k,i}}}$ is deleted from the database.
    \end{enumerate}

\subsection{Managing Internal Datagrams}\label{sec:internal}

In this section, we propose the idea of labeling IP datagrams that are used internally for network management. These datagrams should not be analyzed by the \gls*{ids} and will undergo an internal security procedure. This optimizes the scheme by minimizing the number of \gls*{ids} operations on data packets. The advantage of this protocol is the use of the Identification field, flags and fragment offset as the embedding positions in the IP datagram header which will appear random-like and will not show an evident pattern that an attacker may try to exploit (cf. Section~\ref{sec:threat}). The algorithm for managing IP datagrams by network nodes is as follows: \\

\begin{center}
    \begin{algorithm}[!htbp]
    \caption{(Internal managing algorithm at source node):}\label{alg6}
    \hspace*{\algorithmicindent} \textbf{input:}  IP datagram $d_{IP}$ \\
    \hspace*{\algorithmicindent} \textbf{output:}  Embedding hash value
    \begin{algorithmic}[1]
    \Procedure{Internal managing embedding process:}{}
    \State \textbf{if}($d_{IP}$ = internal managing packet) \textbf{then}
    \State \hspace{10pt} Compute H(Destination IP $||$ First 20 bytes of $d_{IP}$)
    \State \hspace{10pt} $d_{IP}$(header) $\gets$ H(Destination IP $||$ First 20 bytes of $d_{IP}$)
    \State \textbf{else} 
    \State \hspace{10pt} perform watermark generation and embedding procedure
    \State \textbf{end if}
    \EndProcedure
    \end{algorithmic}
    \end{algorithm}
\end{center}

    \begin{algorithm*}[t]
    \caption{(Internal managing algorithm at destination node):}\label{alg7}
    \hspace*{\algorithmicindent} \textbf{input:}  IP datagram $d_{IP}$ \\
    \hspace*{\algorithmicindent} \textbf{output:}  Require IDS/internal-managing
    \begin{algorithmic}[1]
    \Procedure{Internal managing process:}{}
    \State Receive ($d_{IP}$)
    \State \textbf{if}(IP datagram $\gets$  ($IP_{src.addr}$,$IP_{dest.addr}$)) \textbf{then}
    \State \hspace{10pt} \textbf{if}(($IP_{src.addr}$, $IP_{dest.addr}$ = internal) \& L($R_{D}$) = L($D$)) \textbf{then}
    \State \hspace{20pt} Compute H(Destination IP $||$ First 20 bytes of $d_{IP}$)
    \State \hspace{20pt} Extracted H $\gets$ $d_{IP}$(Identification+Flags+Offset)
    \State \hspace{20pt} \textbf{if}(Computed H = Extracted H) \textbf{then}
    \State \hspace{40pt} $d_{IP}$ is authenticated as ``internal-managing data packet"
    \State \hspace{40pt} $d_{IP}$ is is not examined by \gls*{ids}
    \State \hspace{20pt} \textbf{else}
    \State \hspace{40pt} attack detection
    \State \hspace{40pt} $d_{IP}$ is discarded
    \State \hspace{20pt} \textbf{end if}
    \State \hspace{10pt} \textbf{else}
    \State \hspace{20pt} $d_{IP}$ must be examined by the \gls*{ids}
    \State \hspace{10pt} \textbf{end if}
    \State \textbf{else} 
    \State \hspace{10pt} $d_{IP}$ must be examined by the \gls*{ids}
    \State \textbf{end if}
    \EndProcedure
    \end{algorithmic}
    \end{algorithm*}

\begin{figure*}[t]
\begin{center}
    \includegraphics[scale=0.33]{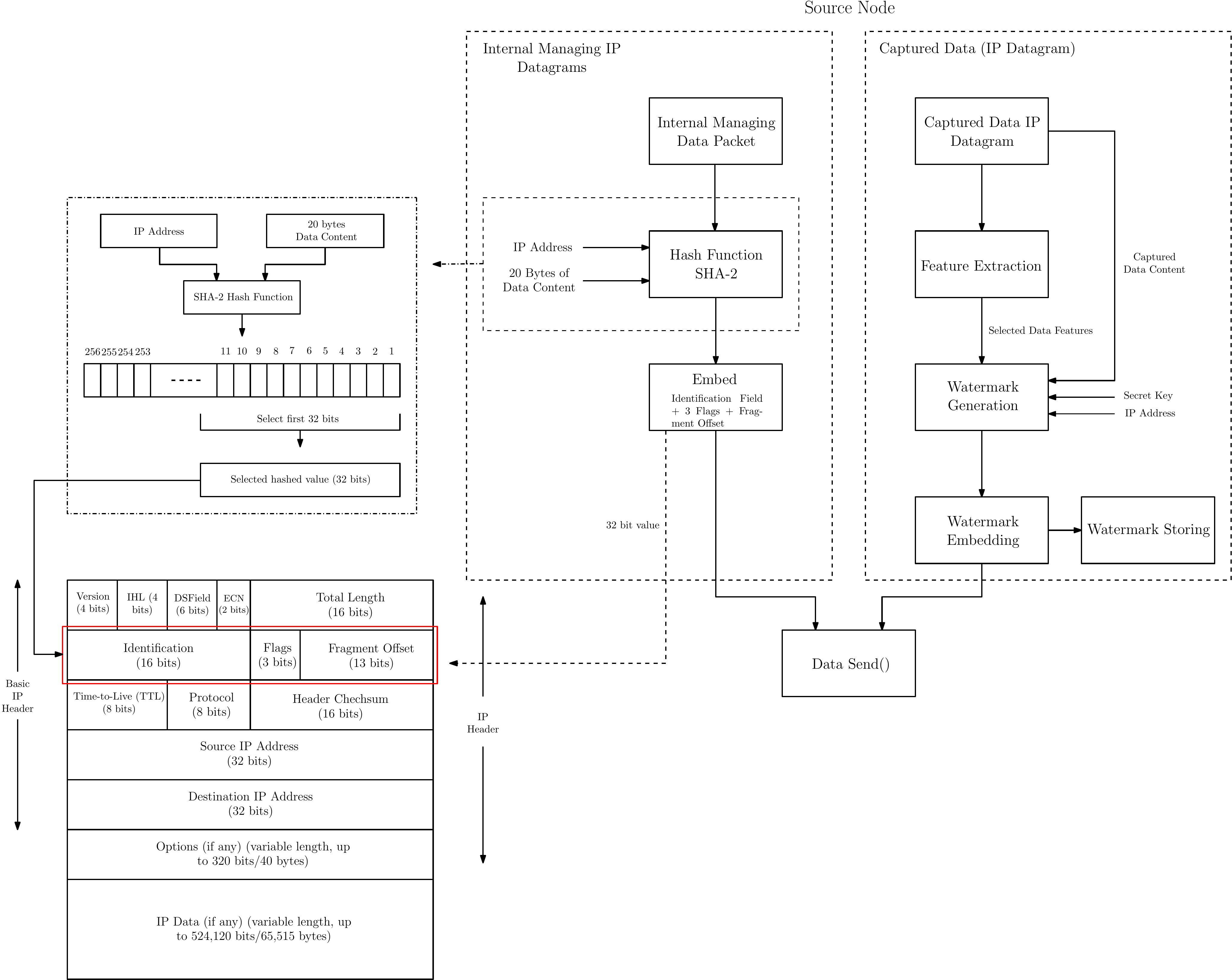}
\end{center}
\caption{Embedding hash value in Internal Managing Protocol.}
\label{fig5}
\end{figure*}

At each node or gateway, internal managing packets are labeled with a hash value that is computed and embedded before the packet is sent. The hash value is computed as H(Destination IP $||$ First 20 bytes of the datagram content), where ``$||$" concatenation. The value is then embedded in the IP datagram header as shown in Algorithm~\ref{alg6} and Figure~\ref{fig5}. We use the Identification field (16 bits), Flags(3 bits) and Fragment offset(13 bits) to embed the selected 32 bit from the hash value. After receiving any IP Datagram at the Gateway or any node in the internal network, an internal managing protocol is performed (before any \gls*{ids} procedure) as shown in Algorithm~\ref{alg7}. The datagram is subjected to a first condition that checks whether these datagrams have both a source and a destination address in our local network, since this is a first condition (filter). Then it checks if both the source and the destination address are internal and the size of the received data packet is equal to an internal managing packet size. Sensed data packets by sensor nodes are watermarked and have different size (data packet + watermark) as shown in~\ref{eq:3}. If it is not the case, the datagram must be examined by the \gls*{ids} (it is not an internal data packet). However, if both IP addresses are internal and the size is confirmed, the device computes $H$(Destination IP address $||$ First 20 bytes of the datagram content) and extracts the hash value embedded in the header of the data packet. Then the node compares these two hash values. If these values are the same, the datagram is authenticated as ``authorized internal-managing packet". Otherwise, an attack is detected and the datagram is discarded.

The hash function used in obtaining the IP datagram label is SHA-2. SHA-2 takes an input of any size and produces a 256-bit hash value. Since the Identification field is 16 bit long and the size of Flags and offset take another 16 bits, making a total of 32 bits, the device selects 32 \gls*{lsb} bits from the hash value as shown in Figure~\ref{fig5}. We can also randomize the selection of these 32 bits by using pseudo-random number generator and obtain randomized bit positions that can be selected. This randomization would add another level of security for the system.

\section{Security Analysis}\label{sec:security}
The \gls*{iot} network can be subject to two main security breaches in the transmission phase: passive and active attacks on both data and watermark. An adversary can launch various attacks based on the threat model described in Section~\ref{sec:threat}. In this section, we provide an analysis for the security of our proposed scheme against the attacks detailed in Section~\ref{sec:threat}. We assume that the network gateway and database are trusted and cannot be compromised by an attacker. 

\begin{theorem}
An unauthorized party cannot access or obtain the secret information generated by the source node $S_n$.
\end{theorem}

\begin{proof}
The source node $S_n$ generates a final watermark $W_{F_{n,k,i}}$ by concatenating two sub-watermarks $sw_{f_{n,k,i}}$ and $sw_{h_{n,k,i}}$. The first sub-watermark $sw_{f_{n,k,i}}$ is obtained from extracted data features as follows: IP address $w_{ip}$ and sensed data capturing time $w_{t}$. These data features are encrypted using \gls*{aes} algorithm using a symmetric secret key $K_j$. $sw_{f_{n,k,i}}$ is generated and encrypted as $E(sw_{f_{n,k,i}}) = ENC(sw_{f_{n,k,i}},K_j)$. Thus, an attacker being unaware of $K_j$ cannot decrypt $sw_{f_{n,k,i}}$ (only authorized parties are aware of $K_j$, \ie intermediate nodes and gateway). Note that $K_j$ is changed and redistributed after a short random number of watermark generation sessions (see Section~\ref{sec:threat}). For the second sub-watermark $sw_{h_{n,k,i}}$, a source node uses a one-way cryptographic hash function $H()$ to obtain $sw_{h_{n,k,i}}$ used for data integrity check. It is computationally infeasible to to find a pair $(x,y)$ such that $h(x) = h(y)$, which make the function secure and cannot be inverted as assumed in Section~\ref{sec:threat}. Additionally, we use SHA-2 hash function in our scheme with 256 bit hash value, which make it computationally infeasible for an attacker to carry out $2^{128}$ calculations to find the second sub-watermark. The generated watermark is computed as $W_{F_{n,k,i}} = E(sw_{f_{n,k,i}})$ $ || $ $sw_{h_{n,k,i}}$ for each captured data. Thus, an adversary cannot access watermark information generated by source nodes.      
\end{proof}

\begin{theorem}
An attacker, cannot successfully deceive an intermediate node $I_l$ or gateway $G$ by inserting fake data or deleting data from the data flow generated by a legitimate node $S_n$ and transmitted to $I_l$ or $G$.
\end{theorem}

\begin{proof}
In case an attacker inserts fake data into a watermarked data-packet $d_{(n,k)W_{F_{n,k,i}}}$ being transmitted to $I_l$ or $G$, the destination node extracts $d_{(n,k)W_{F_{n,k,i}}}$ into sensed data $d_{n,k}$ and watermark $W_{F_{n,k,i}}$. Then, a re-generated sub-watermark $R(sw_{h_{n,k,i}})$ is computed from the received captured data $d_{n,k}$ and compared to the extracted watermark $sw^{'}_{h_{n,k,i}}$ from $W_{F_{n,k,i}}$. The process of re-generation is based on the previously mentioned generation process (\ie SHA-2 hash function for $sw_{h_{n,k,i}}$). Hence, any change in $d_{n,k}$ content produces an altered re-generated sub-watermark. The assumption of secure communication of extracted data features and provenance information using symmetric cryptography and a one-way hash function applies (Section~\ref{sec:threat}). Then, even if an attacker inserts fake data into $W_{F_{n,k,i}}$ without altering $d_{n,k}$, $R(sw_{h_{n,k,i}})$ will not match $sw^{'}_{h_{n,k,i}}$ in the comparison process. Also, if the attacker inserts fake data to the second sub-watermark $E(sw^{'}_{f_{n,k,i}})$ the next hop intermediate node or gateway queries the stored provenance record $p_{n,k,i} = E(sw_{f_{n,k,i}}$ from the data base and compares it with the extracted sub-watermark $E(sw^{'}_{f_{n,k,i}})$. Any change in $E(sw^{'}_{f_{n,k,i}})$ yields to alternation in the provenance information.

In the second case, the attacker aims to delete data content from $d_{n,k}$ or $W_{F_{n,k,i}}$, or drop an entire data-packet $d_{(n,k)W_{F_{n,k,i}}}$ being routed from $S_n$ to $I_l$ or from $I_l$ to $G$. The deletion of $q$ bits from $d_{n,k}$ results in the modification of $R(sw_{h_{n,k,i}})$ and thus $sw^{'}_{h_{n,k,i}}$ will not match $R(sw_{h_{n,k,i}})$. Again, the previously mentioned assumption of secure communication of $W_{F_{n,k,i}}$ applies. Furthermore, if the attacker deletes $q$ bits from $W_{F_{n,k,i}}$ it will be detected in the comparison process of the two sub-watermarks $R(sw_{h_{n,k,i}})$ and $sw^{'}_{h_{n,k,i}}$ or between the queried sub-watermark $E(sw_{f_{n,k,i}})$ and the extracted one $E(sw^{'}_{f_{n,k,i}})$. Obviously, such an adversary may drop $d_{(n,k)W_{F_{n,k,i}}}$ routed through $I_l$. This attack can be detected at $G$ by accessing the tamper-proof database and querying the provenance records of $d_{n,k}$, and detecting where the packet drop attack occurred. The database stores provenance records securely, which cannot be accessed by an attacker (as described in Section~\ref{sec:threat}).  
\end{proof}

\begin{theorem}
An attacker, attempting to alter provenance information: $(i)$ cannot add legitimate nodes to the provenance of data generated by an unauthorized node, $(ii)$ cannot successfully add or remove nodes from the provenance of data generated by legitimate nodes.
\end{theorem}

\begin{proof}
$I_l$ stores a provenance record $W_{F_{n,k,i}}$ after checking data integrity and provenance of the received data-packet $d_{(n,k)W_{F_{n,k,i}}}$. The symmetric secret key $K_j$ shared between legitimate nodes is used to obtain the generated watermark $W_{F_{n,k,i}}$ used in data integrity and provenance validation. An unauthorized node generates watermarks using its own secret key that cannot match a generated watermark at $I_l$ using $K_j$. As stated in Section~\ref{sec:threat}, the source node, the intermediate node and the gateway share secret-keys to be used in different steps of the algorithms (encryption/decryption). These keys are changed and redistributed between legitimate nodes after a random number of sessions. Thus, in order to add a legitimate node, an attacker needs to obtain the same symmetric secret key that is only shared within legitimate network nodes of internal registered IP addresses. In the case of two malicious nodes $I_m$ and $I_v$ attempting to execute an attack, a captured data-packet $d_{n,k}$ by a legitimate source node $S_n$ is routed through $S_m$. $d_{n,k}$ has a provenance record of ($I_1$, $I_2$, $I_4$). The malicious node $I_m$ aims to remove $I_2$ and replace it with $I_v$. To add $I_v$ as a provenance record to the database, the malicious node needs to compute the next-hop watermark which requires, as mentioned above, the knowledge of $K_j$ and hash function variables. Hence, the provenance integrity check at the next $I_j$ will fail and an attack is detected. Thus, $I_m$ will fail to add or remove any provenance record from network database. Moreover, provenance records ($W_{F_{n,k,1}}$, $W_{F_{n,k,2}}$, ..., $W_{F_{n,k,i}}$) of a data-packet $d_{n,k}$ are stored in a tamper-proof database that is assumed to be resistant to any alternation of its entities, attackers cannot alter any record stored in it (see Section~\ref{sec:threat}).
\end{proof}

\begin{theorem}
It is impossible for an attacker, whether acting alone or in collaboration with others, to add or authenticate nodes to the provenance of data produced by a compromised node.
\end{theorem}

\begin{proof}
An attacker may generate fake data and store provenance information in the database as a legitimate node with its secret key. The packet is then forwarded to the next hop intermediate node to store the next hop provenance information in the set of provenance records $P_{n,k}$ for this data packet in the database. The attacker's aim is to construct the provenance from innocent forwarding nodes and make them responsible for false data forwarding, thus marking them as untrustworthy nodes. However, there is an integrity and provenance validation procedure at the next hop node, which includes a watermark re-generation process $W_{F_{(n,k,i)}}$ using the  secret key $K_j$, the attacker do not know the key for legitimate nodes. Thus, this attack will fail at the first hop. 
\end{proof}

\begin{theorem}
Any unauthorized attempt to modify data content through transmission channel would be detected.
\end{theorem}

\begin{proof}
An adversary may perform a modification to the embedded watermark (computed as $W_{F_{n,k,i}} = E(sw_{f_{n,k,i}})$ $ || $ $sw_{h_{n,k,i}}$) or data elements $d_{(n,k)W_{F_{n,k,i}}}$. If data elements are modified and $W_{F_{n,k,i}}$ remains unchanged, a different watermark is obtained based on a wrong hash value at an intermediate node or gateway. Since the first generated sub-watermark at source node $sw_{h_{n,k,i}}$ is the output of a hash function SHA-2 obtained as $sw_{h_{n,k,i}} = H(d_{n,k})$. Again, the assumption of hash functions used in the system (Section~\ref{sec:threat}) applies. The wrong re-generated sub-watermark $R(sw_{h_{n,k,i}})$ will not match the extracted sub-watermark $sw^{'}_{h_{n,k,i}}$. The intermediate node or gateway detects the modification attack and discards the data. Furthermore, if the attacker modifies $W_{F_{n,k,i}}$ and the data payload remains unchanged, the intermediate node or gateway re-generates the right sub-watermark, extract the modified watermark from the received data packet and queries the provenance record $p_{n,k,i}$ from the database. This results in a failed comparison operation for data integrity or for provenance validation and data will be discarded.
\end{proof}

\begin{theorem}
By including a timestamp in the generation process of watermarks, any fraud transmission of previously captured data packets will be discovered.
\end{theorem}

\begin{proof}
An attacker may provide a false idea about the sensing environment by fraudulently transmitting previously heard data packets that are captured and transmitted by a legitimate source node~\cite{Roosta2008}. The attacker also detects the timing characteristics to be used later during the packet replay attack. To deceive an intermediate node or gateway, the attacker updates the timestamp $w_{t}$ of the heard data packet $d_{n,k}$, based on timing characteristics, to a new recent time value. In the proposed scheme, a source node generates a watermark $W_{F_{n,k,i}}$ for each data packet captured ($d_{n,k}$). The generation process is based on provenance information, a timestamp and a hash value as described in Equation~\eqref{eq:1}. Provenance information and timestamp will be encrypted using a secret key $K_j$ to form the first sub-watermark (\ie encrypting $sw_{f_{n,k,i}} = w_{ip} $ $||$ $w_{t}$ as $E(sw_{f_{n,k,i}}) =$ $ENC(sw_{f_{n,k,i}}$,$K_j)$). At next hop $I_l$ or $G$, a new sub-watermark is generated from the replayed packet that will be compared to the extracted sub-watermark. If the attacker changed the timestamp of the data packet $d_{n,k}$ the comparison operation will fail. Since timestamps are different the new re-generated sub-watermark will not match the extracted one. Note that the attacker cannot modify the timestamp $w_{t}$ embedded in the watermark $W_{F_{n,k,i}}$, due to the encryption process performed on the generated sub-watermark $sw_{f_{n,k,i}}$. The sub-watermark is encrypted using the source secret key $K_j$, which is only shared with legitimate entities (intermediate node and gateway) where an attacker uses a different secret key as stated in Section~\ref{sec:threat}. Hence, replaying an old packet with an updated timestamp will lead to a failed authentication procedure.
\end{proof}

\begin{theorem}
In Algorithm 7, an attacker trying to deceive network devices to accept malicious datagrams as trusted internal managing datagrams will be detected and examined by implemented security algorithms.
\end{theorem}

\begin{proof}
If an attacker succeeds to modify an internal managing datagram, the datagram will be forwarded to the next hop node. At the receiving node, a hash value is computed from the content of the datagram using a one way hash function as detailed in Section~\ref{sec:internal}. It also extracts the hash value embedded by the source node from the identification field of the IP header. Both hash values are then compared to detect any attempt of forgery attack. If the values do not match, the device applies the implemented security algorithms to the received IP datagram and an attack procedure is performed. Note that a source node uses a one-way cryptographic hash function $H()$ using SHA-2 to obtain the hash value (embedded in internal managing IP datagram's  header) so that it is computationally infeasible to find a pair $(x,y)$ such that $h(x) = h(y)$, making it impossible for an attacker to invert the hash value and embed it to deceive the system (Section~\ref{sec:threat}). Hence, a malicious entity trying to deceive the forwarding nodes using internal managing datagrams will be detected and discarded.
\end{proof}

Based on the above security analysis, the proposed scheme is proven to be resistant against various malicious attacks of \gls*{iot} networks, such as modification attack, integrity attack, packet replay, database authentication attack and passive attacks. It guarantees the integrity of data and ensures security against identifying and retrieving provenance information in \gls*{iot} networks. 

\section{Simulation Results and Analysis}
\label{sec:results}
In this section, we evaluate the performance of the proposed scheme based on two features: data integrity and data provenance. For data integrity, the proposed scheme is evaluated based on watermark generation, embedding and verification time. Also, we have measured how this scheme performs in terms of energy utilization. The results are then compared to three state of the art techniques: \gls*{rwfs}~\cite{Alromith2018}, \gls*{act}~\cite{Sun2013} and \gls*{zwt}~\cite{hameed2018} based methods. We selected these three state-of-the-art methods to assess the performance of our new security technique based on their use of different security techniques deployed in a similar network model. For data provenance, we compare our scheme with MAC-based provenance scheme (MP) in terms of cost analysis. The algorithms were implemented in MATLAB\textsuperscript{\texttrademark} on Intel core i7 processor with a 2.59 GHz clock cycle and 16 GB of memory. Sensor data is represented as an integer data type, since most sensor readings are of numeric form such as temperature, humidity, motion and intensity. \par
In our algorithm, we use \gls*{aes} with 128 bit key size for encryption of generated watermarks. Despite the fact that \gls*{aes} has a larger key size than \gls*{des}, \gls*{aes} is a more secure and advanced encryption algorithm compared to \gls*{des}, which makes it more resistant to cryptanalysis attacks. Another reason for using \gls*{aes} is its performance and efficiency. \gls*{aes} is a fast and efficient algorithm. We provide, in Figures~\ref{fig:f6},~\ref{fig:f7},~\ref{fig:f8} and~\ref{fig:f9}, a comparison of using \gls*{aes} and \gls*{des} algorithms in the generation and verification processes at each sensor node in our proposed model. The results show the better performance of our scheme when applying \gls*{aes} algorithm (approximately 10 times faster) in both watermark generation and verification. The use of substitution-permutation network (SPN) structure, which is optimized for hardware implementation and allows for parallel processing in \gls*{aes} shows a faster performance than \gls*{des}, which uses a Feistel network structure. Regarding the hash function, we use a one way hash function SHA-2, specifically SHA-256, for generating the second sub-watermark $sw_h$. Although SHA-1 is faster than SHA-2 functions since it uses a smaller block size and has a simpler construction, however it is important to note that the slower performance of SHA-2 functions is outweighed by their improved security compared to SHA-1. We compared the generation and verification time of our proposed model using different hash functions in Figures~\ref{fig:f10} and~\ref{fig:f11}. The results shows that SHA-1 is faster than SHA-2 functions and SHA-2(256) function requires less processing time than SHA-2(384) and SHA-2(512). Hence, we use SHA-2(256), which provides the best performance in SHA-2 functions, as our hash function in the generation of watermarks. 

\begin{figure*}[!htbp]
  \centering
  \subfloat[]{\includegraphics[trim = 11mm 5mm 21mm 2mm,clip,width=0.5\textwidth]{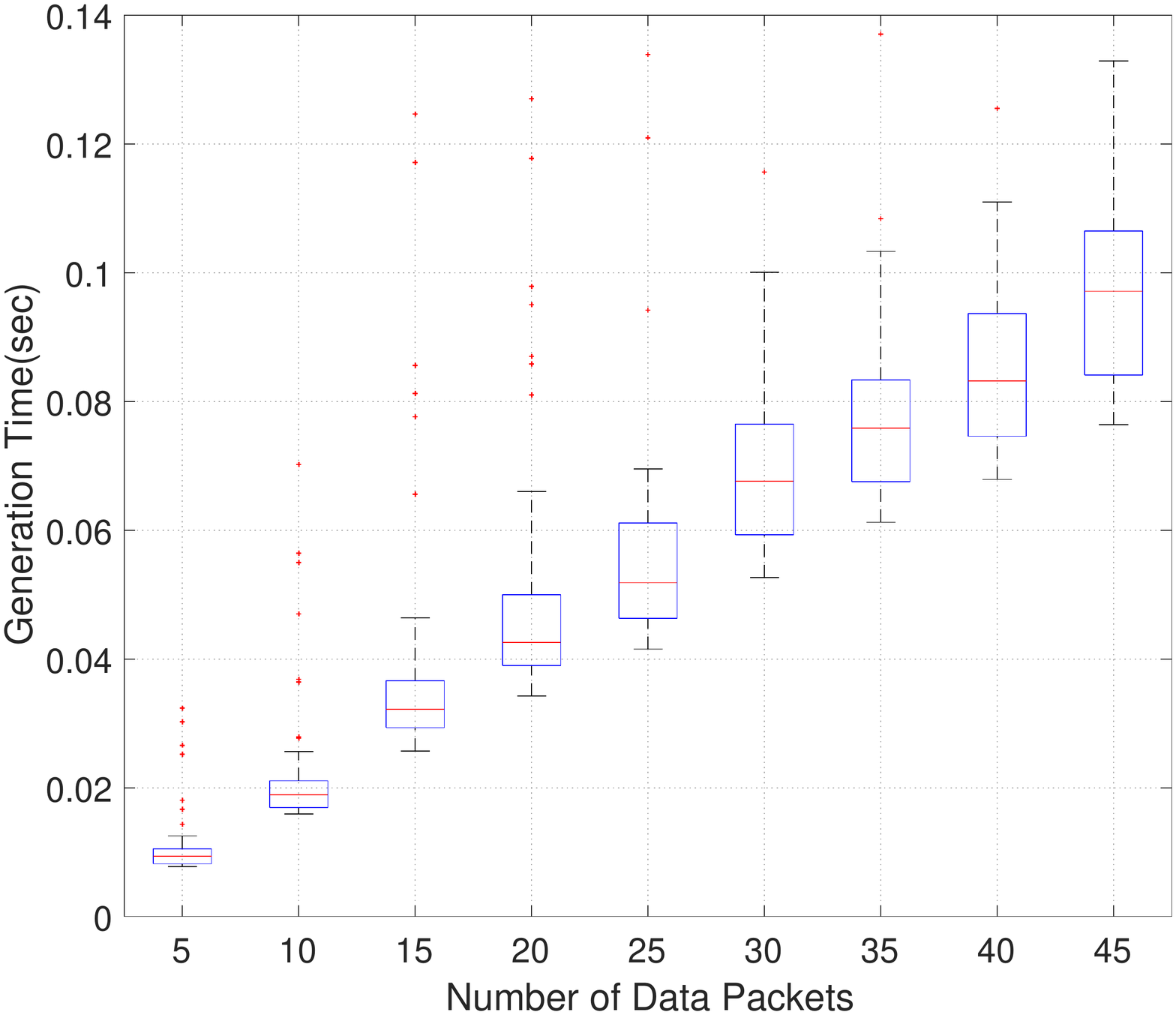}\label{fig:f6}}
  \subfloat[]{\includegraphics[trim = 11mm 5mm 21mm 2mm,clip,width=0.5\textwidth]{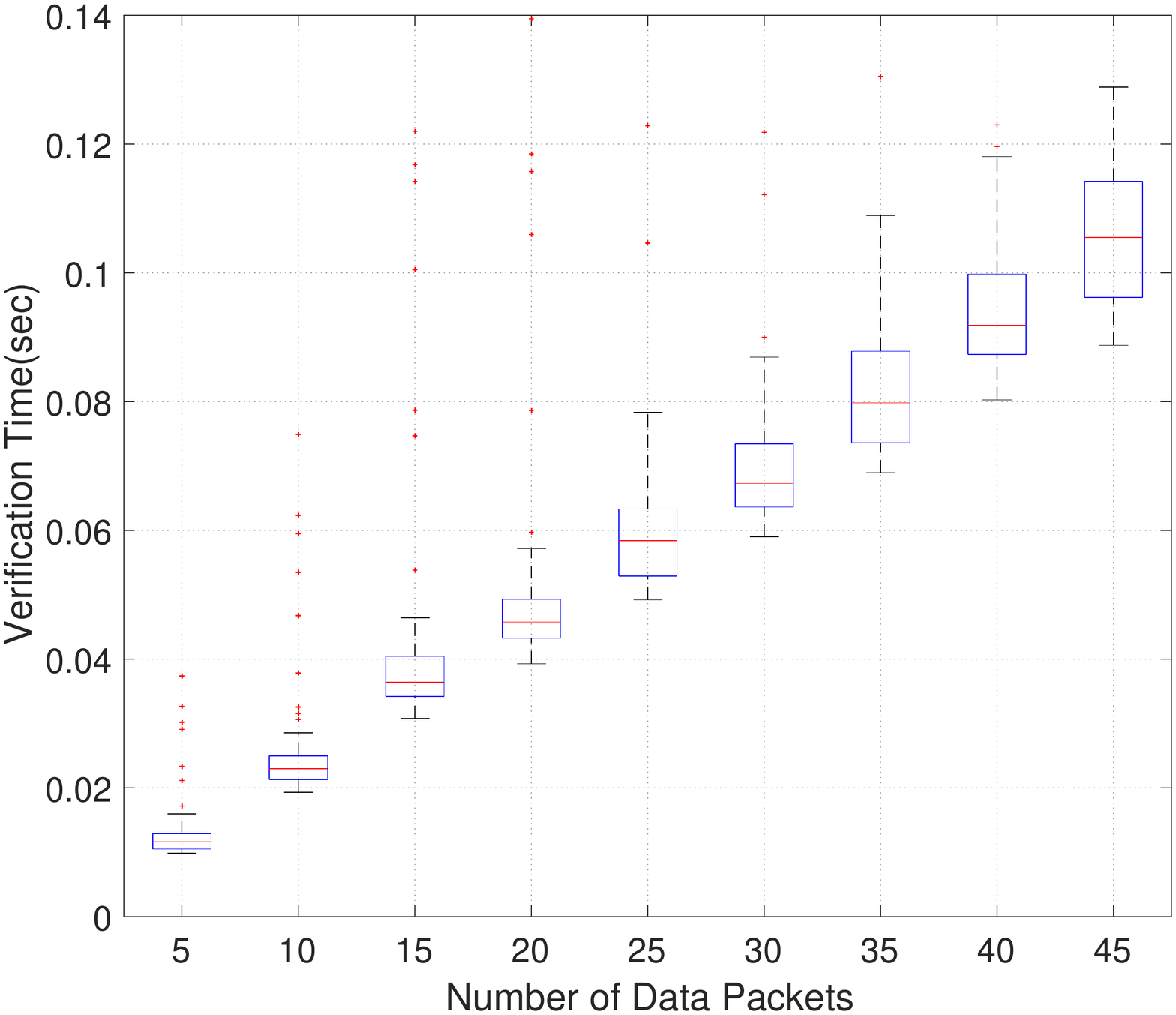}\label{fig:f7}}
  \caption{Computational time. (a) Watermark generation and embedding time using AES. (b) Watermark verification time using AES.}
\end{figure*}

\begin{figure*}[!htbp]
  \centering
  \subfloat[]{\includegraphics[trim = 11mm 5mm 21mm 2mm,clip,width=0.5\textwidth]{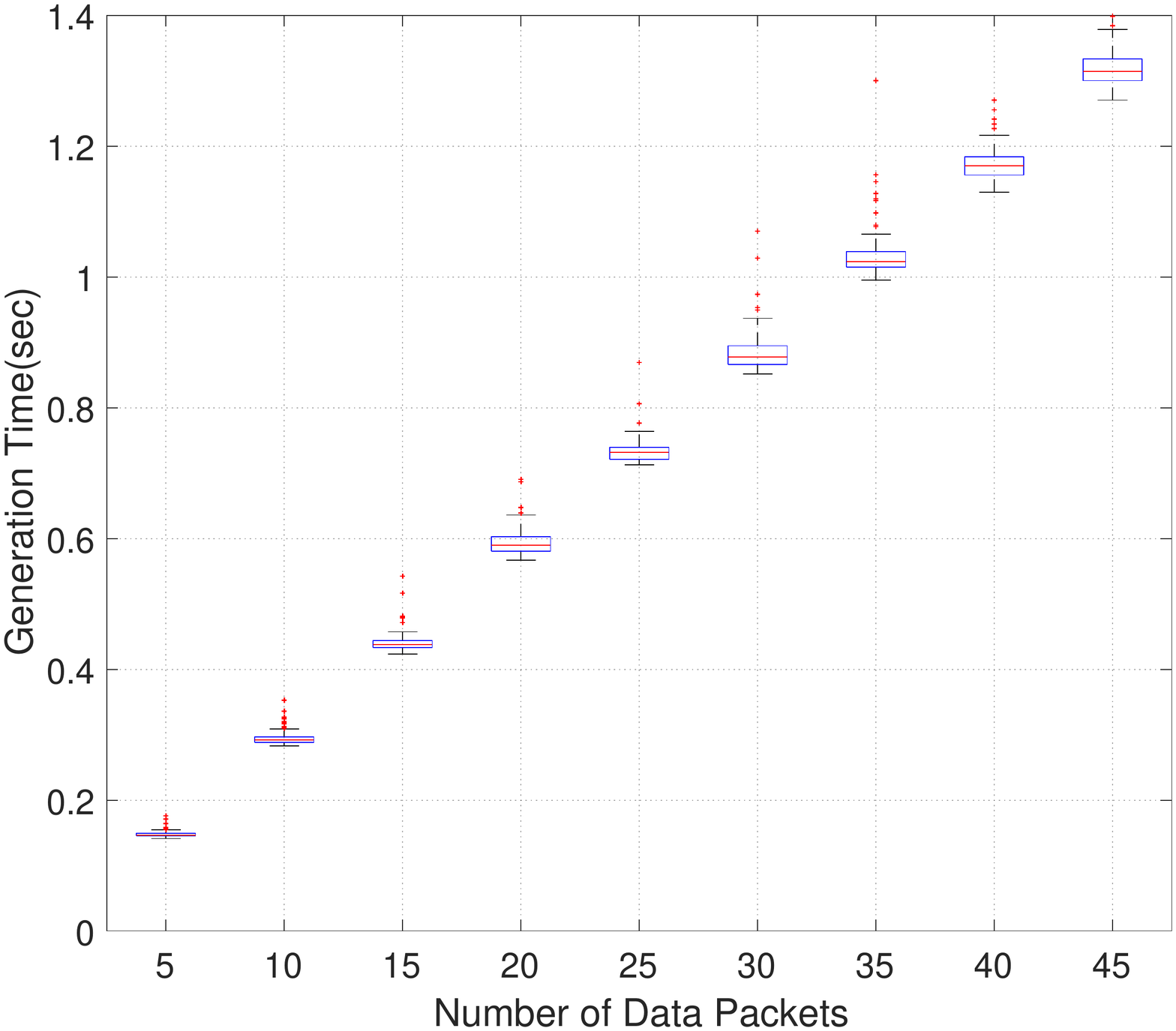}\label{fig:f8}}
  \subfloat[]{\includegraphics[trim = 11mm 5mm 21mm 2mm,clip,width=0.5\textwidth]{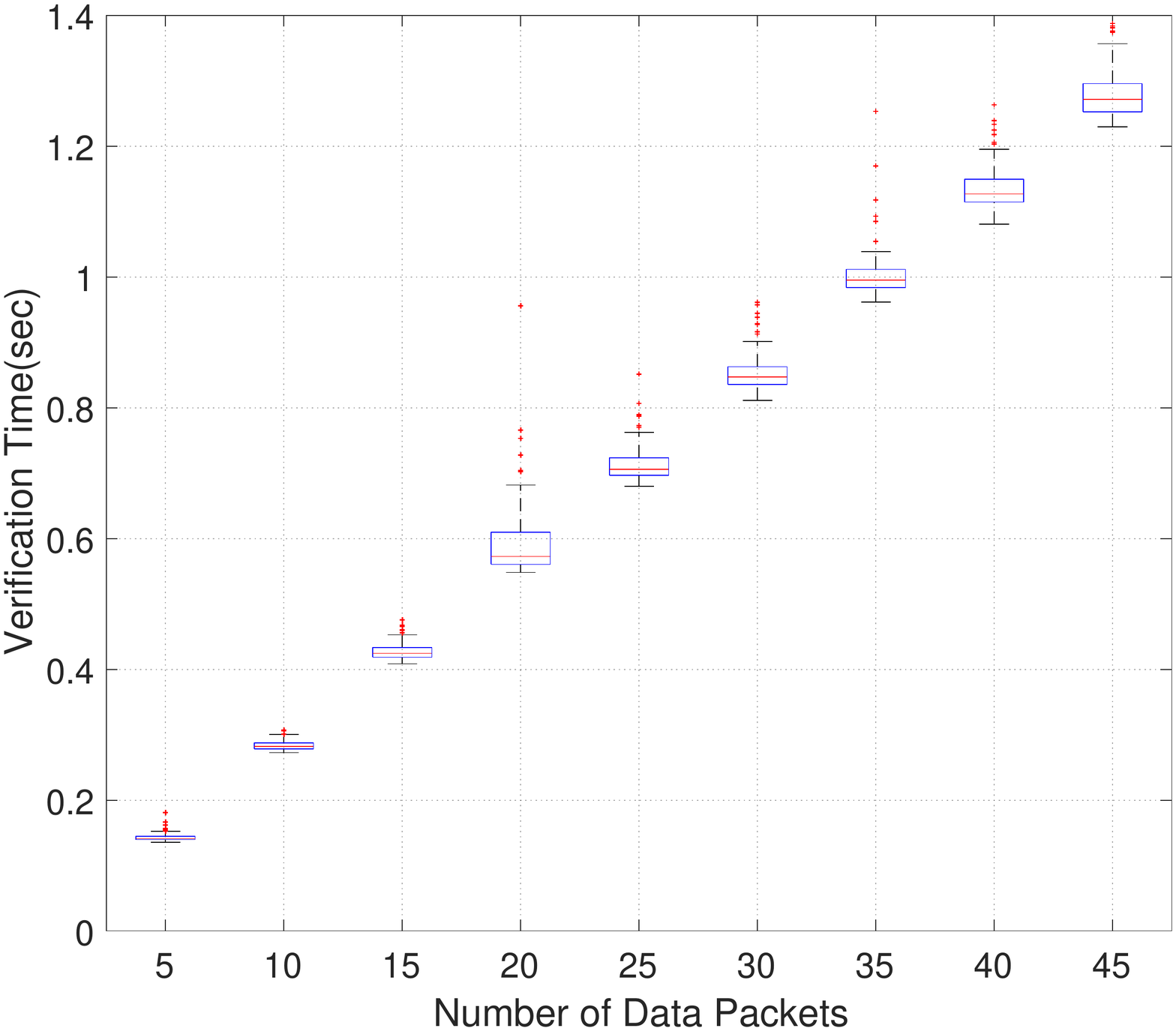}\label{fig:f9}}
  \caption{Computational time. (a) Watermark generation and embedding time using DES. (b) Watermark verification time using DES.}
\end{figure*}

\begin{figure*}[!htbp]
  \centering
  \subfloat[]{\includegraphics[trim = 11mm 1mm 21mm 0mm,clip,width=0.5\textwidth]{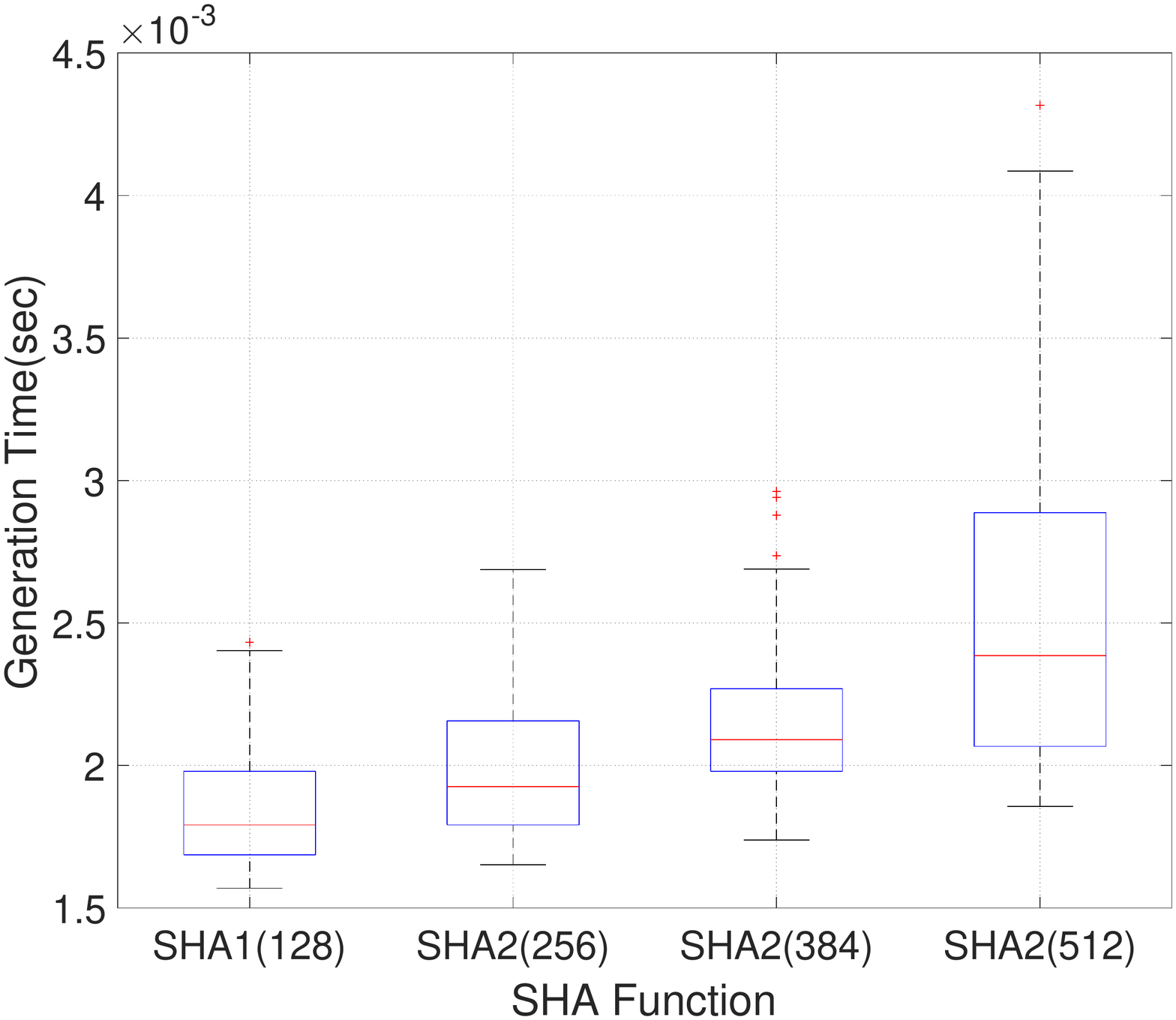}\label{fig:f10}}
  \subfloat[]{\includegraphics[trim = 11mm 1mm 21mm 0mm,clip,width=0.5\textwidth]{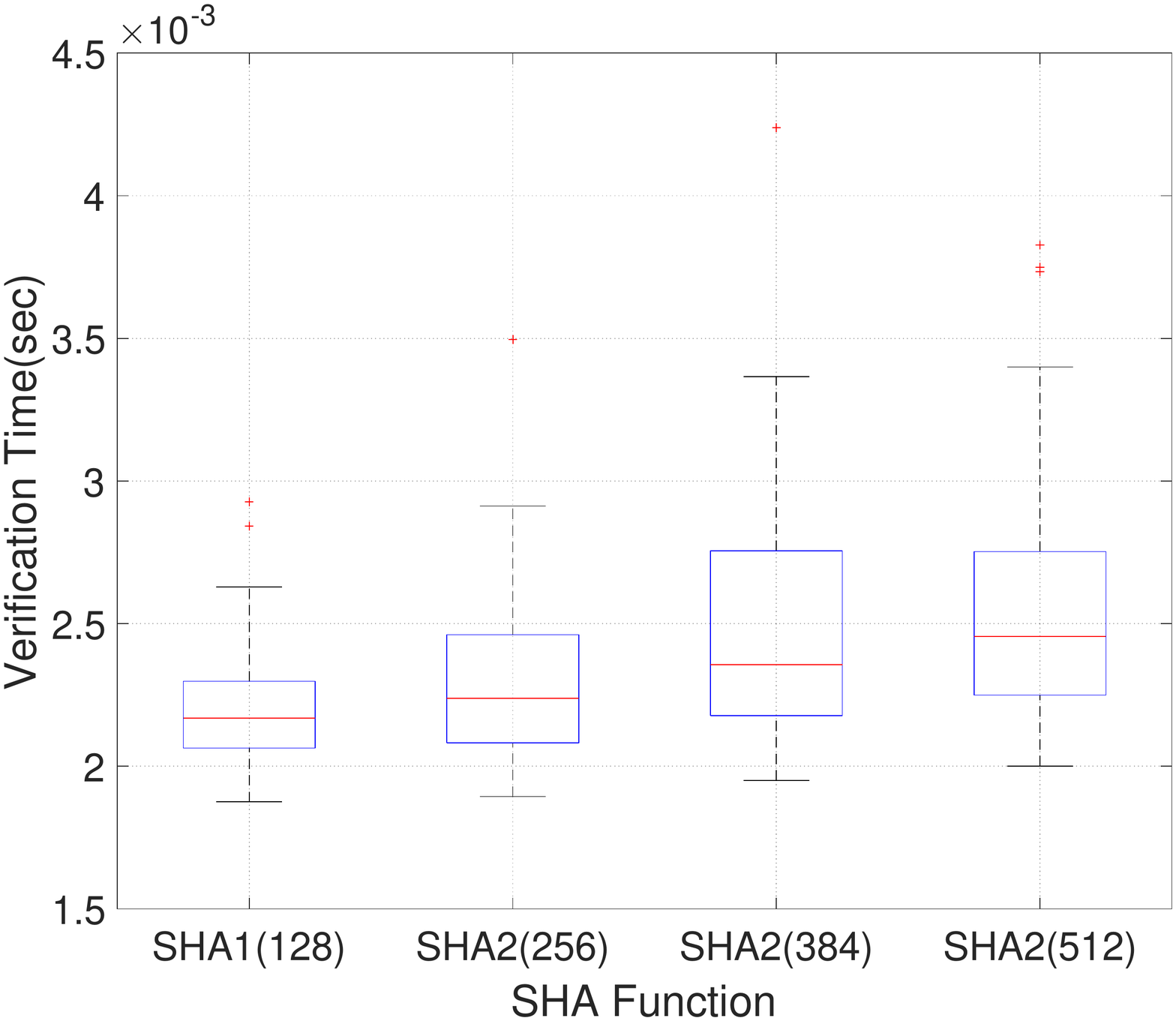}\label{fig:f11}}
  \caption{SHA Comparison. (a) Watermark generation and embedding time using different SHA functions. (b) Watermark verification time using different SHA functions. }
\end{figure*}

\subsection{Performance Evaluation}
To evaluate the performance, we measure the computational time such as watermark generation, embedding, and watermark verification time of our proposed scheme, \gls*{rwfs}~\cite{Alromith2018}, \gls*{act}~\cite{Sun2013} and \gls*{zwt}~\cite{hameed2018}. Additionally, we used energy utilization as another performance metrics and compared the results with existing methods~\cite{Alromith2018,Sun2013,hameed2018}. We select these three works from the literature to compare our model with a regular watermarking technique, asymmetric cryptography technique and a zero-watermarking technique. From our research work, these papers provide these three methods and deploys it in a scenario similar to what we are analyzing and studying. Note that a confidence interval is added to show the average generation and verification time after a 100 simulation runs.

\subsubsection{Computational time}
Computational time is described as the time required to complete the following processes: watermark generation, embedding and verification at sensor nodes and gateway. 

\begin{enumerate}
\item \textbf{Watermark generation and embedding time:}  It is the time taken by a source node or intermediate node to generate a watermark and embed it in the data packet. The existing \gls*{rwfs}~\cite{Alromith2018} generates a watermark by encrypting the sensed data with a homomorphic encryption algorithm proposed by Castelluccia et al. \cite{Castelluccia2009} and passing it as an input to a keyed-hash message authentication code (HMAC). The watermark is then embedded randomly by computing each position of watermark bits using a pseudo-random number generator (PRNG) for each captured data at source node. In \gls*{act}~\cite{Sun2013}, the watermark generation is based on an asymmetric cryptography function and uses group hashing for a set of data values that need to be captured in different time intervals before generating the watermark. Additionally, \gls*{zwt}~\cite{hameed2018} uses \gls*{des} for watermark encryption in the watermark generation process. Comparing these approaches~\cite{Alromith2018,Sun2013,hameed2018}, our proposed scheme uses a zero-watermarking technique that generates a fixed size watermark from provenance information and data features. It applies a one-way hash function to extracted data features and symmetric encryption (\ie \gls*{aes}) for provenance information. Simulation results shows that the proposed scheme requires less watermark generation and embedding time than existing approaches~\cite{Alromith2018,Sun2013,hameed2018} as observed in Figure~\ref{fig:f12}. Using AES as an encryption and SHA-2 to generate watermarks shows a significant improvement in the performance of sensor nodes. This results in decreasing the end-to-end time from capturing data to processing it at the destination gateway. \\
\item \textbf{Watermark verification time:} The verification algorithm is used to extract watermark and verify data integrity at the destination node. This procedure is performed at an intermediate node or gateway which have more computational and power capabilities than source nodes. In the proposed approach, the watermark is concatenated to the data payload and each watermark is generated using AES and SHA-2 for each data packet which requires less extraction and verification time than \gls*{rwfs}~\cite{Alromith2018}, \gls*{act}~\cite{Sun2013} and \gls*{zwt}~\cite{hameed2018}. The time for extracting and verifying data integrity in~\cite{Alromith2018} depends on computing each watermark bit position and computing a hash value after re-encrypting the extracted data. In \cite{Sun2013}, the intermediate node or gateway requires receiving several data packets to perform watermark extraction and re-calculating the watermark based on asymmetric encryption to perform verification. Moreover, in \gls*{zwt}~\cite{hameed2018}, the intermediate node or gateway needs to extract data features from the received data packet and encrypt these features using DES algorithm to re-generate the watermark for verification. Figure~\ref{fig:f13} shows that the proposed zero-watermark approach requires less time to extract and verify data integrity than existing schemes~\cite{Alromith2018,Sun2013,hameed2018}. It is worth pointing out that the proposed approach provides both data integrity and data provenance. The time shown in Figure~\ref{fig:f13} for our proposed scheme includes also the time needed for querying the stored watermarks from the database.  
\end{enumerate}

\begin{figure*}[!htbp]
  \centering
  \subfloat[]{\includegraphics[trim = 11mm 5mm 21mm 2mm,clip,width=0.5\textwidth]{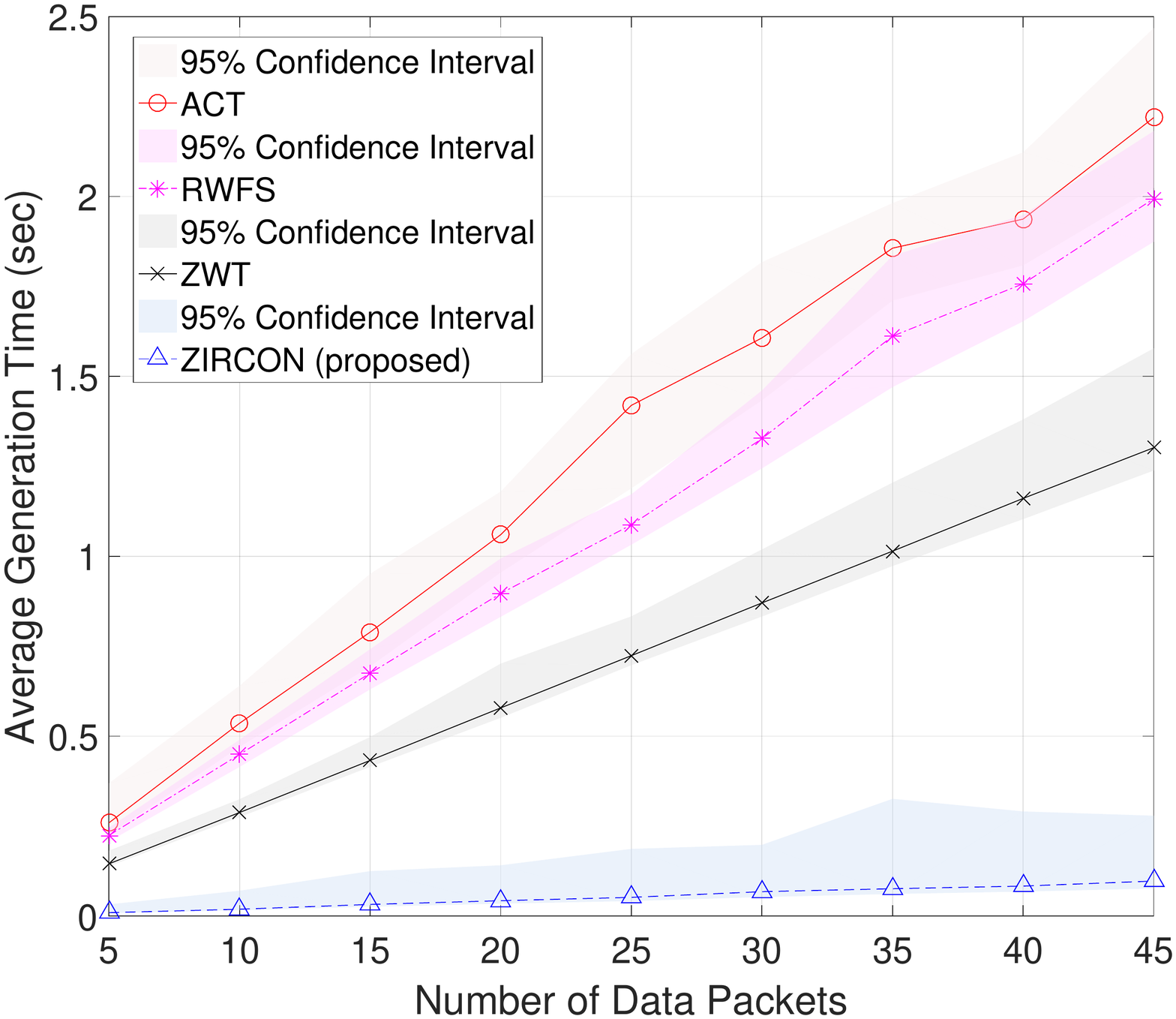}\label{fig:f12}}
  \hfill
  \subfloat[]{\includegraphics[trim = 11mm 5mm 21mm 2mm,clip,width=0.5\textwidth]{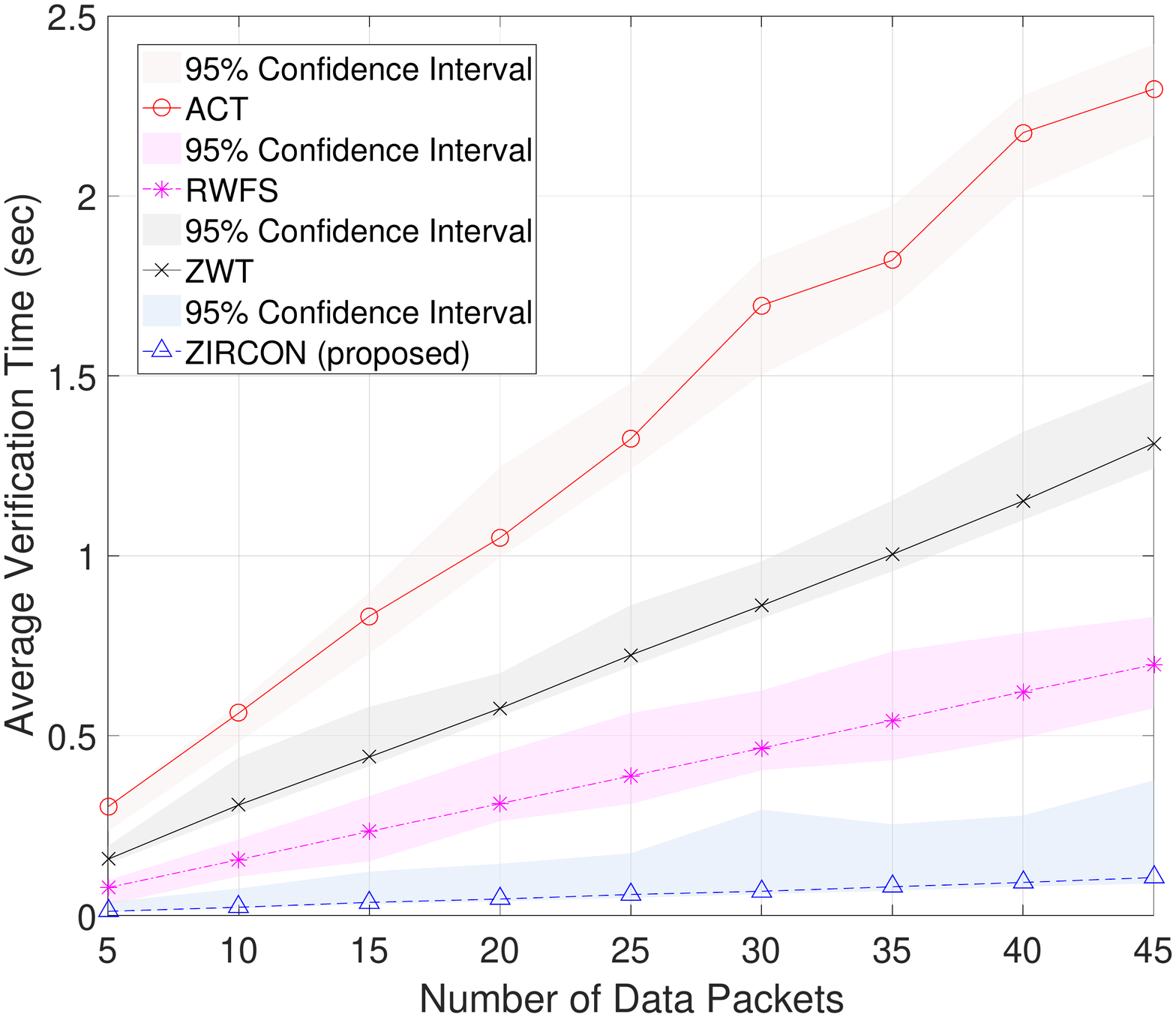}\label{fig:f13}}
  \caption{Computational time. (a) Watermark generation and embedding time. (b) Watermark verification time.}
\end{figure*}

\subsubsection{Energy consumption} Energy consumption evaluates the energy consumed by a sensor node from the power utilized by each node and the total time consumed in the sensor node operation steps as shown in Figure~\ref{fig14}. The energy consumed by a sensor node varies based on several basic energy consumption sources: processing time cost, radio transmission, sensor sensing, transient energy, and sleeping time cost~\cite{Halgamuge2009,Fu2013,Malka2009}. It is crucial to utilize less energy-consuming security mechanisms for \gls*{iot} networks due to the limited computation and power capabilities of sensor nodes. In the proposed scheme, we made our assumptions regarding energy consumption due to the fixed space required for watermark embedding. The phases that affect energy consumption in a sensor node are sensor node activation cost, watermark generation and embedding cost, data capturing cost, data transmission cost and cost for going to sleeping mode. The energy ($E_n$) of each sensor node in the network is computed according to Equation~\eqref{eq:3}. The power ($P_n$) utilized by each node is determined by node's hardware components, the network's data rate, and the communication protocols used by the network. In order to estimate the power consumption of the sensor node for numerical simulation we use the energy model in~\cite{Miller2005} based on Mica2 Motes. The time to complete a round of a sensor node operation specified in Figure~\ref{fig14} is $T_n$ which varies according to the data processing method and functionality of this node as shown in the figure. We assume, as in~\cite{Miller2005}, that the parameters used in the energy calculation are as follows: $T_{A}$ = 1ms (Active time cost), $T_{S}$ = 0.5 ms (Data sensing time), $T_{C}$ (Computation and processing time), $T_{TR}$ = 300 ms (Data transmission time), $T_{SL}$ = 299 ms (sleeping time cost), and $P_n$ = 30mW (Average power consumption of a single sensor node). Using Equation~\eqref{eq:5} we compute the energy of sensor nodes based on the previously specified parameters.

    \begin{equation} \label{eq:4}
    E_{n} = P_{n} \times T_{n}
    \end{equation}

    \begin{equation} \label{eq:5}
    E_{n} =  P_{n} \times (T_{A} + T_{S} + T_{C} + T_{TR} + T_{SL})
    \end{equation}
    
The analysis of the energy consumption of ZIRCON scheme compared to \gls*{rwfs}~\cite{Alromith2018}, \gls*{act}~\cite{Sun2013} and \gls*{zwt}~\cite{hameed2018} shows that our approach requires less energy for each operating node. This results in an increase in life time of our network compared to other networks. The higher energy consumption in \gls*{rwfs}~\cite{Alromith2018} is based on the computation of bit positions for watermark embedding and the encryption of captured data (that is used as an input to an HMAC function to obtain a final watermark) using homomophic encryption algorithm. This method is slower than conventional symmetric encryption methods because of the complex mathematics it requires. In comparison to homomorphic encryption, symmetric encryption is quicker and easier to use because uses a single key to encrypt and decrypt data. Also, in \gls*{act}~\cite{Sun2013} the use of asymmetric cryptography functions and group hashing requires more energy at each sensor node due to the additional computational overhead required for the public and private key operations. The existing scheme \gls*{zwt}~\cite{hameed2018}, which uses \gls*{des} for watermark encryption, dissipates higher energy than our proposed scheme which uses \gls*{aes} for sub-watermark encryption. Figures~\ref{fig15a} and \ref{fig15b} shows the energy consumption of our proposed scheme compared to existing state-of-the-art methods \gls*{rwfs}~\cite{Alromith2018}, \gls*{act}~\cite{Sun2013} and \gls*{zwt}~\cite{hameed2018}, for a single source node and an intermediate node respectively. It is clearly shown that ZIRCON requires less energy consumption at each node of the network. 

\begin{figure*}[!htbp]
\centerline{\includegraphics[width=0.65\textwidth]{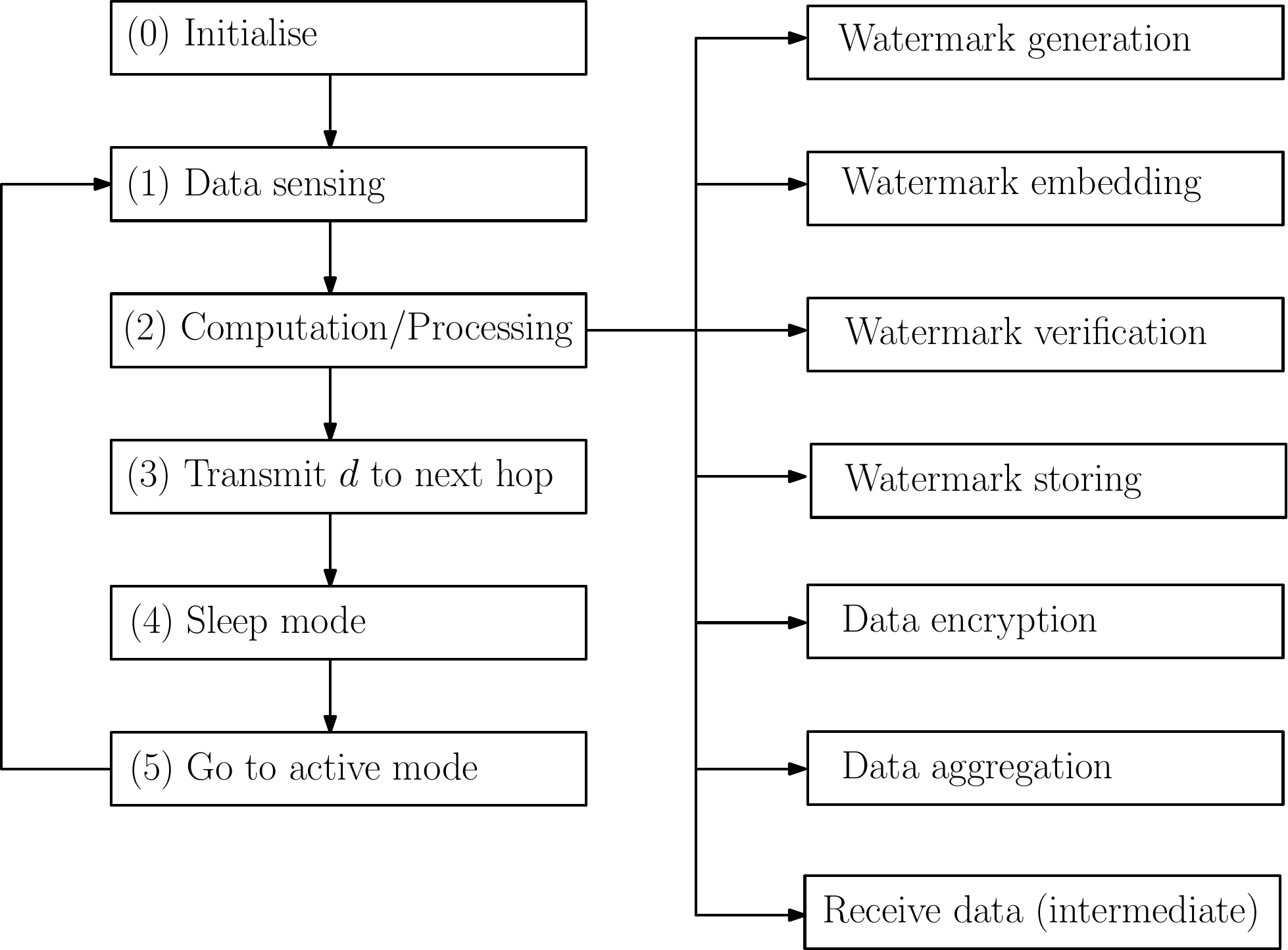}}
\caption{Sensor node operation cycle.}
\label{fig14}
\end{figure*}

\begin{figure*}[!htbp]
\centerline{\includegraphics[width=0.75\textwidth]{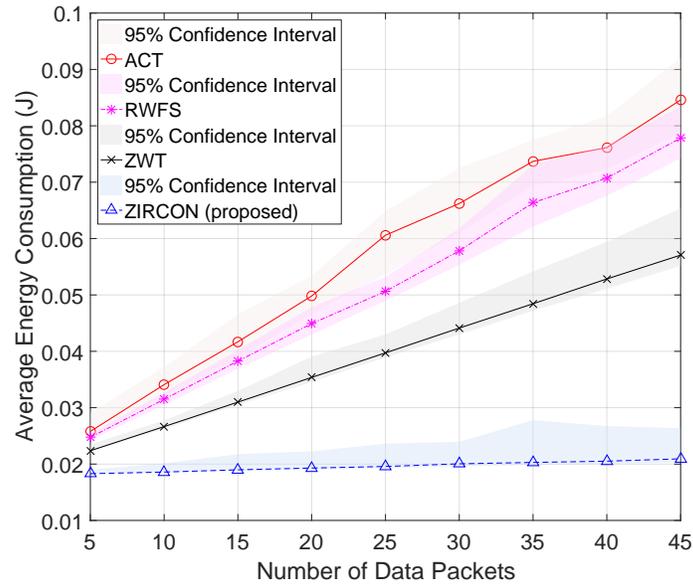}}
\caption{Energy consumption cost per single source node.}
\label{fig15a}
\end{figure*}

\begin{figure*}[!htbp]
\centerline{\includegraphics[width=0.75\textwidth]{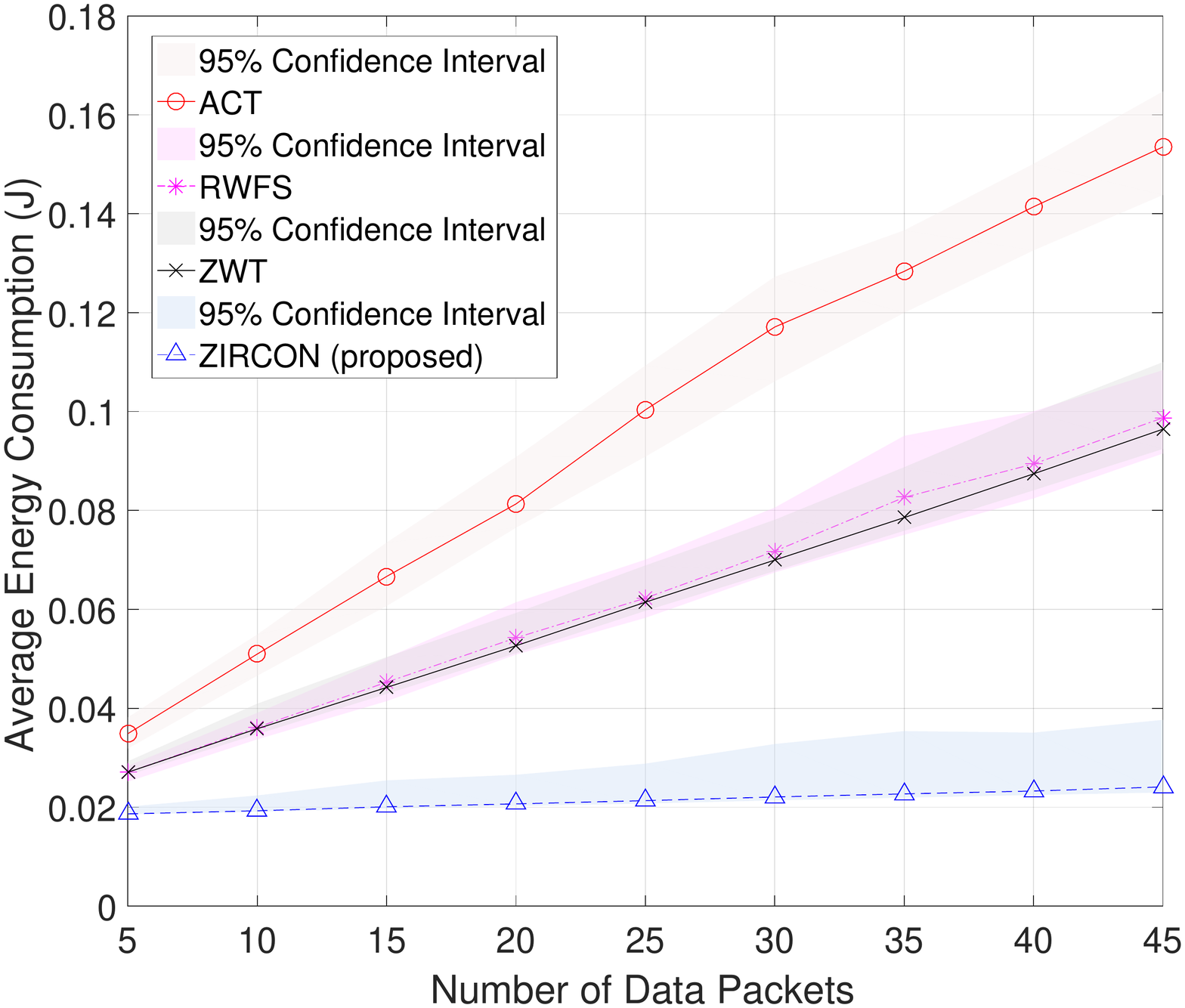}}
\caption{Energy consumption cost per single intermediate node.}
\label{fig15b}
\end{figure*}

\subsection{Cost Analysis} Regarding cost analysis, we compare ZIRCON with three state-of-the-art methods in terms of transmission data size and data packet length. These approaches are as follows: 

\begin{enumerate}
    \item The secure provenance framework $SProv$ \cite{ragib2009} that is adapted to sensor networks by \cite{Salmin2015}. The provenance record at a node $n_i$ is $p_i=<n_i, \mathrm{hash}(D_i),C_i>$, where $\mathrm{hash}(D_i)$ is a one way hash function of the updated data and $C_i = \mathrm{sign}(\mathrm{hash}(n_i,\mathrm{hash}(D_i)|C_{i-1})$ is an integrity checksum. This method is referred to as SSP. 
    \item The MAC-based provenance scheme which computes a MAC value and send it with the node ID as the provenance record. This method is referred to as MP \cite{Salmin2015}.
    \item A lightweight secure scheme BFP \cite{Salmin2015} that uses Bloom Filters to encode provenance information, which is sent along the data path with the data packet. 
\end{enumerate}

In our proposed scheme, a sensor node transmits both provenance information (IP address and timestamp) and a hash value as a zero-watermark. The IP address and the timestamp both have a size of 4 bytes. The source node encrypts the sub-watermark $sw_{f_{n,k,i}}$ and produce an encrypted sub-watermark of 16 bytes. Also, the source node computes the hash value from the extracted data payload, as shown in Algorithm~\ref{alg1}, and selects the first 8 bytes. This implies that the generated zero-watermark including the provenance record is 24 bytes. The provenance record is stored in a tamper-proof network database at each hop. Hence, each data packet holds only one generated zero-watermark in each hop. For SSP, to perform cryptographic hash operations, they utilize SHA-1 with a bit length of 160, and for generating digital signatures of 160 bits (ECDSA), they make use of the TinyECC library~\cite{Liu2008}. The node ID, which is 2 bytes long, results in each provenance record being 42 bytes in length. To implement MP, the provenance record is formed of node ID and a MAC value computed on each source node. It uses the TinySec library~\cite{Karlof2004} to compute the sensor CBC-MAC of size 4 bytes. Thus, the provenance record is of 6-byte size. In both schemes, SSP and MP, each node embeds its provenance information as a record with data packet, as the path length increases the provenance size increases linearly. This increase in provenance leads to an increase in the transmitted data packet size. In a multi-hop scenario, the provenance is $6 \times H$ bytes (\ie the path is formed of $H$ hops) for MP and $42 \times H$ bytes for SSP. However, in BFP, the provenance length depends on parameter selection of the Bloom Filter. For a given $H$ and a false positive probability $P_{fp} = 0.02$, the number of required bits to encode the provenance information is $m = (-H\cdot \ln(P_{fp}))/(\ln 2)^2$. In this case the length of a BF grows with the number of nodes. Figure~\ref{fig:f16} shows a comparison between ZIRCON, SSP, MP and BFP approaches in-terms of transmission data size in a single hop scenario. Similarly, the results for data packet length in a Multi-hop scenario for both schemes is shown in Figure~\ref{fig:f17}. In resource constrained networks, energy is mainly affected by data transmission, which increases as the data packet increase. The results show that ZIRCON performs better than SSP and MP as the number of hops increases in the sensor network. Also, our algorithm outperforms BFP in terms of provenance length and scalability as the size of the network increases, and as the number of hops exceeds 11. Our proposed model only encodes one provenance record $p_{n,k,i}$ with each data packet $d_{n,k}$ during transmission.

\begin{figure*}[!htbp]
  \centering
  \subfloat[]{\includegraphics[trim = 40.5mm 85mm 40.5mm 90mm, clip,width=0.5\textwidth]{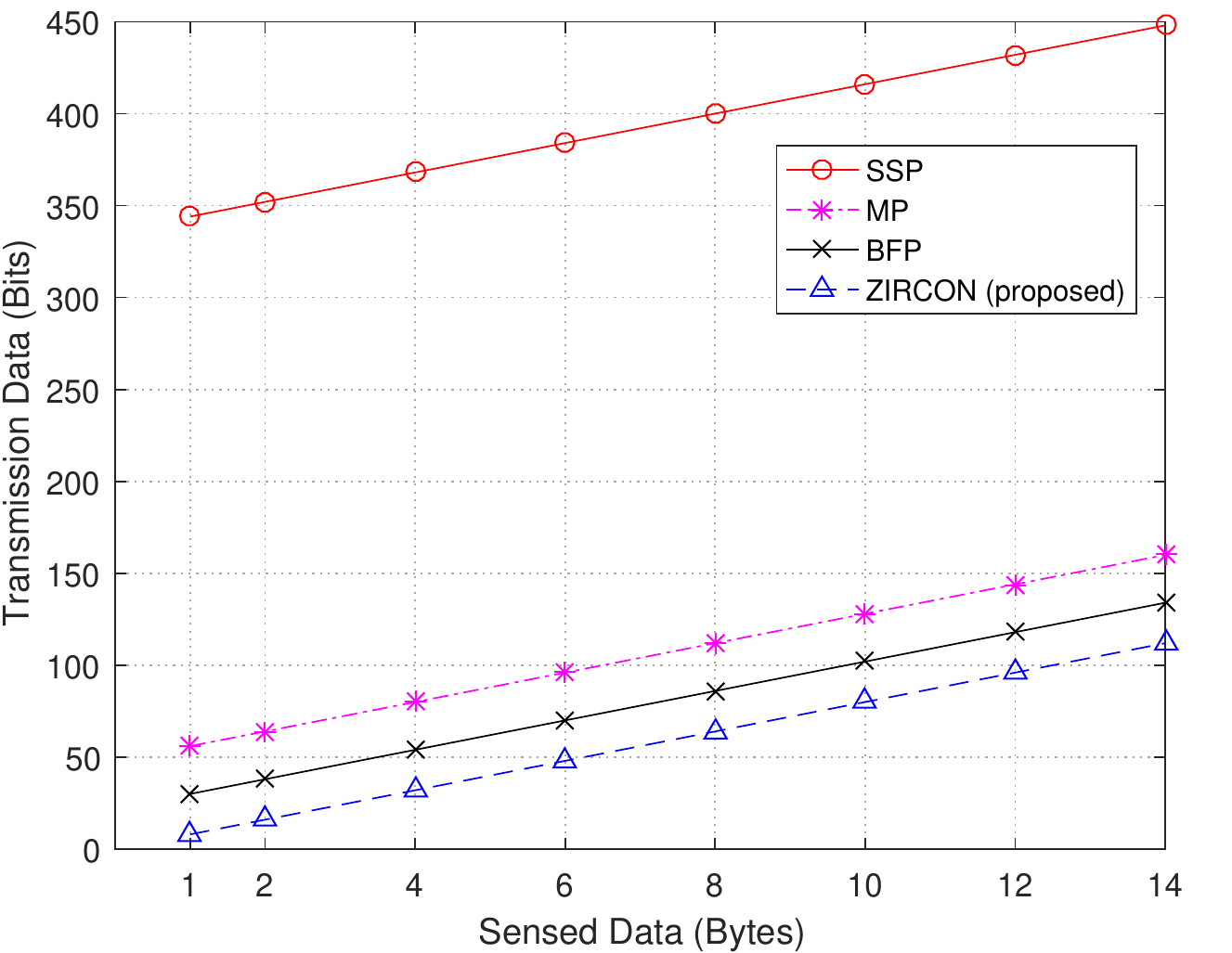}\label{fig:f16}}
  \hfill
  \subfloat[]{\includegraphics[trim = 40.5mm 85mm 40.5mm 90mm, clip,width=0.5\textwidth]{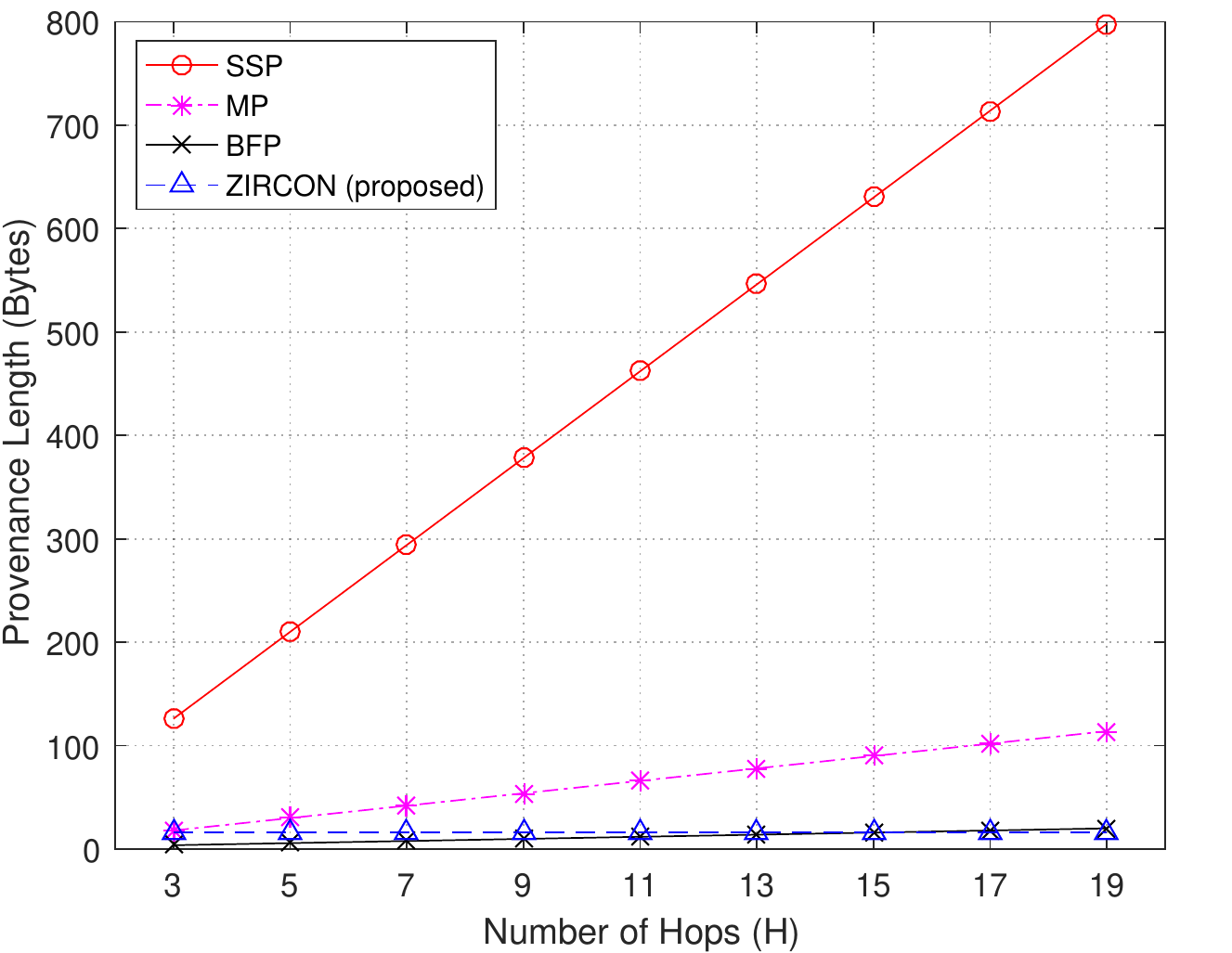}\label{fig:f17}}
  \caption{Cost comparison. (a) Transmission data size in single hop scenario. (b) Provenance length in Multi-hop scenario.} 
\end{figure*}

\section{Discussion}
\label{sec:discussion}

The related literature includes many proposed schemes for ensuring data integrity and secure provenance transmission in \gls*{wsn}s using digital watermarking. These models are elaborated in section~\ref{sec:relatedWork}. The limitations of such solutions were addressed in our proposed scheme. In this scheme we combine both data integrity and secure provenance transmission, taking into consideration the computational capabilities of sensor nodes in \gls*{iot} networks, while maintaining security standards. \gls*{iot} networks are vulnerable to many type of attacks. These networks are used in decision making processes that require high level of security. Moreover, it is essential in many situations to keep track of the data captured from sensor nodes to identify any malicious traffic in the context of intrusion detection systems. For this, it requires to overcome a set of challenges in order to securely transmit provenance information. The difficulties involve handling processing overhead of each network node, transmitting information about the origin of data in an efficient way without using extra bandwidth and quickly responding to any security breaches. Provenance information grows very fast, which requires transmitting large amount of provenance information with data packets. In fact, building the lineage of each data-packet requires storing the information of the data-packet including the complete set of nodes that were covered from source to destination. Embedding such vast amount of information with the data packet will result in a massive network overhead. This requires a solution for handling this amount of provenance information. This critical problem was not addressed in the related literature. In this context, we propose the use of a tamper-proof database to store these information that are embedded in watermarks at each node covered in the network. Hence, to obtain the required security standards, the proposed zero-watermarking approach generates two sub-watermarks that are used for integrity verification and secure provenance transmission. The sub-watermarks are based on one-way hash function (\ie SHA-2) and symmetric encryption (\ie \gls*{aes}). In our work, we provide an efficient and secure way to keep track of the whole network route that a piece of information has taken despite bandwidth overhead, storage limitations and computational overhead, while ensuring data integrity.

Managing internal data packets is a key component of improving \gls*{ids} efficiency. Analyzing each data packet by the \gls*{ids} at each node implies additional computational overhead. This issue was not addressed in the related literature, which only focus on data packets that are specified for sensed data. In our model, we propose a protocol for labeling internal managing data packets which allows to check for any attack at the level of these packets without the need to analyze it by \gls*{ids}. In our work, we validate our zero-watermarking algorithm through a security analysis that shows our approach is robust to many attacks based on an attack model. Additionally, we provide a performance evaluation to analyze computational time, energy consumption and cost analysis in comparison with related literature.   

As a result of our security scheme outlined and proposed in this paper, there are several areas for future study and improvement. The fast evolution of security attacks against \gls*{iot} networks, such as \gls*{ddos}, Botnets, Privacy invasion and Physical attack, requires the advancement in security measures to protect against these types of attacks and to ensure the security of \gls*{iot} networks. This presents an important call to discover ways to tackle the vast number of security attacks that are not yet studied in the area of securing data provenance in integrity in \gls*{iot} networks. Another issue is that large-scale \gls*{iot} networks that introduce the problem of large-scale provenance need to be analyzed. How to handle the huge lineage of data being transmitted over long data-path. Even in the presence of a database, how to manage methods to efficiently overseeing a large quantity of sensor nodes and the data they collect, along with information about its origin and the path it covers.

\section{Conclusion}
\label{sec:conclusion}

This paper addresses the problem of data integrity and secure transmission of provenance information for \gls*{iot} networks. We propose a zero-watermarking approach that embeds provenance information and data features with data packets and stores these watermarks in a tamper-proof network database. The security capabilities of ZIRCON makes it secure against different type of sensor network attacks, as proved using a formal security analysis. We have validated our findings by conducting representative simulations and compared our results with existing schemes based on different performance parameters. The results show that the proposed scheme is lightweight, has better computational efficiency and consumes less energy, compared to prior art. Perspectives for future work include testing this scheme to various attacks that are not included in our threat model and study the possibility of using this method for large-scaled provenance.

%\section*{Acknowledgments}
%The authors acknowledge the funding obtained from the EIG CONCERT-Japan with grant PCI2020-120689-2 (Government of Spain) ``DISSIMILAR'', and to the RTI2018-095094-B-C22 ``CONSENT'' and PID2021-125962OB-C31 ``SECURING'' projects granted by the Spanish Ministry of Science and Innovation. The authors acknowledge as well support from IMT’s CyberCNI chair (Cybersecurity for Critical Networked Infrastructures, cf. \url{https://cybercni.fr/}).

\bibliographystyle{IEEEtran}
%\bibliography{references}


\end{document}